\documentclass[11pt]{article}
  \usepackage[left=1in, right=1in, top=2cm]{geometry} 
\usepackage{lineno}
\usepackage[dvipsnames]{xcolor}

\usepackage[utf8]{inputenc}
\usepackage{amsmath, amssymb,multirow} 
\usepackage{amsthm}
\usepackage{natbib}
 \newtheorem{theorem}{Theorem}[section]
 \newtheorem{corollary}{Corollary}[theorem]
 \newtheorem{lemma}[theorem]{Lemma}
\theoremstyle{remark}

\usepackage{accents}
\usepackage{mathtools}
\usepackage{commath}
\usepackage[sc,osf]{mathpazo}
\usepackage{caption,subcaption}
\usepackage{graphicx}
\newcommand{\ut}[1]{\underaccent{\sim}{#1}}
\renewcommand{\vec}[1]{\ut{#1}}
\begin{document}
% \begin{frontmatter}
\title{Non-parametric generalised newsvendor model}
\author{Soham Ghosh  \and Sujay Mukhoti}

% \end{frontmatter}

\maketitle
\begin{abstract}
In classical newsvendor model, piece-wise linear shortage and excess costs are balanced out to determine the optimal order quantity. However, for critical perishable commodities, severity of the costs may be much more than linear. In this paper we discuss a generalisation of the newsvendor model with piece-wise polynomial cost functions to accommodate their severity. In addition, the stochastic demand has been assumed to follow a completely unknown probability distribution. Subsequently, non-parametric estimator of the optimal order quantity has been developed from a random polynomial type estimating equation using a random sample on demand. Strong consistency of the estimator has been proven when the true optimal order quantity is unique. The result has been extended to the case where multiple solutions for optimal order quantity are available. Probability of existence of the estimated optimal order quantity has been studied through extensive simulation experiments. Simulation results indicate that the non-parametric method provides robust yet efficient estimator of the optimal order quantity in terms of mean square error. 

%\keywords{
%Stochastic programming \and  Non-parametric Estimation \and Monte-Carlo Simulation \and Newsvendor Problem  \and Strong Consistency \and Non-linear Optimisation }
  
\end{abstract}
\section{Introduction}
Newsvendor problem deals with determination of optimal order quantity of a perishable commodity by offsetting piece-wise linear shortage and excess costs and without allowing any backlog. The  decision for a single period problem is taken at the beginning, \emph{i.e.} before the random demand is realised
 \citep[see][and the references therein]{chernonog2018}. However, perishable critical resources would often warrant shortage and excess costs to be more severe than linear. For example, chemotherapy drugs are administered to patients as per a schedule. Shortage of the drug on the scheduled day would result in breaking of the treatment cycle. Here the loss is more severe than merely the quantity lost. Similarly, excess inventory of critical drugs or chemical resources might cause  vast environmental and microbial hazards during disposal of the excess material. In this work, we discuss a piece-wise non-linear alternative to the classical newsvendor model to accommodate severity in the decision \citep{mukhoti2021,halman2012approximating}. 

Non-linear newsvendor problem has been studied only recently in the literature. \cite{parlar1992} considered the periodic review inventory problem and derived the solution of a newsvendor problem with a quadratic cost function. \cite{gerchak1997} described optimal order quantity determination from a newsvendor problem with linear excess but quadratic shortage cost.
\cite{pal2015distribution} used exponential weight function of order quantity to the holding cost and linear excess cost in a newsvendor set-up. \cite{kyparisis2018price} addressed the newsvendor problem for quadratic utility function. 
 \cite{khouja1995newsboy,chandra2005}, among others, considered optimisation of reliability function of the stochastic cost. 
 In this paper, we consider generalisation of the classical newsvendor problem by modelling the severity of shortage and excess costs. In particular, we introduce measurable and continuous non-linear weights to the two types of costs and establish the conditions for existence of the optimal order quantity. 
 
A critical issue with the optimal order quantity determination in classical newsvendor problem is the lack of knowledge on random demand. Majority of the works assume a completely specified demand distribution, whereas in reality, it is seldom so. In case of unknown demand distribution, parametric and distribution-free estimation of the optimal order quantity has been considered more recently. Parametric estimation of the optimal order quantity has been studied by \cite{nahmias1994demand} and more recently, \cite{kevork2010estimating} for Normal demand. \cite{agrawal1996estimating} estimated the order quantity for negative binomial demand.  
% \cite{sok1980} in his master's thesis, presented estimators of the optimal order quantity based on order statistics for parametric distributions including uniform and exponential.
\cite{rossi2014confidence} has given bounds on the optimal order quantity using confidence interval for parametric demand distributions. \cite{mukhoti2021} estimated optimal order quantity for uniform and exponential demands in non-linear newsvendor problem. 
% They also studied the impact of power misspecification.

Distribution free estimation of optimal order quantity, on the other hand, has been studied in two parallel ways in the context of classical newsvendor problem. In the first case, the investigator has access to population summary measures like mean, variance etc, but the demand distribution remains unknown \citep{bai2020distributionally}. \cite{scarf1958} and later \cite{moon1994distribution} studied the min-max optimal order quantity in such cases. The second approach considers the estimation problem based on an uncensored random sample from the unknown demand distribution. \cite{pal1996, bookbinder1998} discussed construction of bootstrap based point and interval  estimator of the optimal order quantity using demand data. The sampling average approximation (SAA) method \citep[see][]{shapiro2001,linderoth2006empirical}, replaces the expected cost  by the sample average of the corresponding objective function and then optimises it. \cite{levi2015data} provides bounds of the relative bias of estimated optimal cost using SAA based on uncensored demand data. 
However, not much work has been done on non-parametric estimation of optimal order quantity in non-linear newsvendor problems. 

 In this paper we devise a non-parametric technique to estimate the optimal order quantity in the generalised model. Our study makes two unique contributions to the literature. First, we develop a non-parametric estimator of the optimal order quantity in a generalised newsvendor set-up, which has not been attempted in the literature to the best of our knowledge. The non-parametric estimator is developed from an estimating equation using an uncensored random sample on stochastic demand. The feasibility of obtaining solutions to the estimating equation has been derived in almost-sure sense using its random polynomial representation. We have studied the asymptotic performance of the estimated optimal order quantity. We have shown strong consistency of the optimal order quantity estimator when the true one is unique. We also present the extension of the above strong consistency result in both the cases, where true optimal order quantity is not unique or both true and estimated optimal order quantities are not unique. Next, we have provided a simulation based way to estimate the probabilities of existence of feasible roots of a random polynomial and the distribution of the roots in the generalised newsvendor context. Our results on the properties of the estimated optimal order quantity are based on 3 million simulation experiments for \emph{Uniform} and \emph{Exponential} demand distributions. 
 We compute the optimal order quantities for the two demand distributions and study the properties of the probability distribution of the estimated optimal order quantity. 
 Since the existence of the estimator of optimal order quantity is not guaranteed, we provide a way to use the simulation results for computing the probabilities of their existence for different combinations of severity and cost for a large sample size of 10000. 
 We also present the performance study of the non-parametric estimator, in small and large samples, using the mean square errors. The paper concludes with a discussion on the findings. 

\section{Symmetric Generalised Newsvendor Problem}
% \textbf{Problem Description}
We consider a single-period newsvendor problem where, excess inventory is disposed of at the end of the period with no salvation cost. We assume instantaneous replenishment of order quantities. Our work considers a case where the severity of the excess and shortage are more than the quantity lost ( \emph{i.e} the gap between inventory and demand). We also assume absence of any influencing factors like marketing efforts,promotions,discounts etc. 

Let the stochastic demand be represented by a random variable $X$ with a compact support $\mathcal{X}\subseteq \mathcal{R}^+$ defined over the complete probability space $(\Omega, \mathcal{F}, P)$, where $\mathcal{F}$ is the $\sigma$-algebra over $\Omega$. In this paper we do not consider pre-booking, which in turn implies $0\in \mathcal{X}$. Further, let $C_e~(0<C_e<\infty)$ and $C_s~(0<C_s<\infty)$ be the  excess and shortage costs per unit respectively. Then the cost function in classical newsvendor set-up at an inventory level $Q$ is given by
\begin{equation}
    C(Q,X)=\left\{ 
    \begin{array}{cc}
    C_e(Q-X), & if ~X\leq Q \\
    C_s(X-Q), & if ~X>Q
    \end{array}
    \right.
    \label{stdcost}
\end{equation}
Related stochastic programming problem under the assumption of existence of $E_{\mathbb{G}}[X]$, is given by
\begin{equation}
   \underset{Q\in \mathcal{X}}{argmin}~E_G[C(Q,X)]
    \label{stdnv}
\end{equation}
where $G(\cdot)$ is the induced probability distribution of $X$ defined over the measurable space $(\mathbb{R}^+, \mathcal{B}^+)$, where $\mathcal{B}^+$ is the corresponding Borel-algebra. 
We consider generalisation of quadratic cost function by introducing polynomial weights (in $Q$ and $X$) of degree $m$ (say, $P_{1,m}(Q,X)$ and $P_{2,m}(Q,X)$) to shortage and excess respectively. Degree of the polynomials ($m$) represents the (equal) severity of shortage and excess. If $m=0$, then the problem reduces to classical newsvendor problem. The severity polynomials should satisfy the following properties:
\begin{itemize}
    % \item[(a)] $f_i(Q,X)$ is $\mathbb{G}$-integrable for every $Q \in \mathcal{X}$; $i=1,2$ 
    \item[(a)] for a given $X$, $P_{i,m}(Q,X)$ is continuously differentiable with respect to $Q ~(\in \mathcal{X})$ for $i=1,2$, up to order $m$ 
    \item[(b)] The $m^{th}$ derivative of $P_{i,m}(Q,X)$ is finite, $i=1,2$.
    \item[(c)] If for any  convergent sequence $\{X_n\}$ in $\mathcal{X}$,  $ X_n \overset{a.s.}{\to} Q$, then $P_{i,m}(Q,X_n)\overset{a.s.}{\to} 0$ for $i=1,2$ $(a.s.\; \Rightarrow\; almost\; sure)$.
\end{itemize} 
 Based on the above properties, a natural choice for the severity polynomials are as follows:
\begin{align}
    & P_{1,m}(Q,X)= \sum_{j=0}^{m-1} (-1)^{m-1-j} \binom{m-1}{j}Q^{j}X^{m-1-j}=(Q-X)^{m-1} \label{powerwt1} \\
    & P_{2,m}(Q,X)=\sum_{j=0}^{m-1} (-1)^{m-1-j} \binom{m-1}{j}Q^{m-1-j}X^j = (X-Q)^{m-1} 
    \label{powerwt2}
\end{align}
The constant $m$ is integer valued and $m-1$ could be interpreted as the severity constant. As $m$ increases, more severe is the loss. For $m=1$, no extra severity is implicated and the problem reduces to the classical newsvendor problem.
Thus the new cost function for generalised newsvendor is given by
\begin{equation}
    C_m(Q,X)=\left\{ 
    \begin{array}{cc}
    C_e(Q-X)^m, & if ~X\leq Q \\
    C_s(X-Q)^m, & if ~X>Q
    \end{array}
    \right.
    \label{sygencost}
\end{equation}
The new cost functions could also be interpreted as a generalisation of constant costs per unit ($C_e,C_s$) model to demand and inventory dependent cost models, viz. $C_e (Q-X)^{m-1}$ and $C_s(X-Q)^{m-1}$ respectively.

In view of the above weight function structure, we now make the following assumptions about the probability distribution of demand ($X$):
\begin{itemize}
    \item[A1.]  $\mathcal{X}$ is independent of $Q$
    \item[A2.] $G$ is continuous and strictly increasing over the support $\mathcal{X}$
    % \item[A3.] $G$ is continuously differentiable at all interior points of $\mathcal{X}$
    \item[A3.] $X^m$ is $\mathbb{G}$-integrable $\forall ~m\geq 0$
\end{itemize}
The assumption $A1$ is required to avoid the trivial solution of zero order quantity, which may arise for certain choices of demand distribution, the degree of severity ($m$) and the costs $(C_e,C_s)$. For example, if the demand is $Unif(0,2Q)$ then for $C_e=C_s$, the optimum order quantity would become zero. Hence, we make further assumption of $C_e\neq C_s$.  

The expected cost function in this case can be written as,
\begin{equation}
    \displaystyle E_{\mathbb{G}}[C_m(Q,X)] = \int_{S_Q} C_e(Q-x)P_{1,m}(Q,x) d\mathbb{G} + \int_{S_Q^\prime } C_s(x-Q)P_{2,m}(Q,X) d\mathbb{G}
    \label{expectedcost}
\end{equation}
where $S_Q=\{\omega \in \Omega \mid X(\omega) \in (0, Q) \}$, $S_Q^\prime=\mathcal{X}\setminus S_Q$ and $E_{\mathbb{G}}$ denotes expectation with respect to $\mathbb{G}$.

Differentiating Eq.~\ref{expectedcost} with respect to $Q$ using Leibnitz rule, we get the first order condition for the minimisation problem stated above as follows
\begin{eqnarray}
    & & \frac{\partial   E_{\mathbb{G}}[C_m(Q,X)]}{\partial Q}  =  0 \nonumber \\
    & \Rightarrow & {\int_{S_Q}  C_e(Q-X)^{m-1} d\mathbb{G}}  =  \int_{S_Q^\prime}{ C_s(X-Q)^{m-1}d\mathbb{G}} \nonumber \\
    & \Rightarrow & C_e \int_{S_Q} (Q-X)^{m-1} d\mathbb{G}  =  C_s \left[\int_\mathcal{X} (X-Q)^{m-1} d\mathbb{G} - \int_{S_Q} (X-Q)^{m-1} d\mathbb{G}\right] \nonumber \\
    & \Rightarrow &  \int_{S_Q} (Q-X)^{m-1} d\mathbb{G}  =  \frac{C_s}{\left[ C_e+C_s(-1)^{m-1}\right]} \int_{\mathbb{X}} (X-Q)^{m-1} d\mathbb{G} \nonumber  \\
    & \Rightarrow & {\frac{E_{\mathbb{G}}\left[ (Q-X)^{m-1}\mathbb{I}(S_Q)\right]}{E_{\mathbb{G}}[(X-Q)^{m-1}]}}  =  k_m
    \label{FOC}
\end{eqnarray}
where, $\mathbb{I}(S_Q)$ is an indicator function over the set $S_Q$ and $k_m = {\frac{C_s}{C_e+(-1)^{m-1}C_s}}$. Denoting ${\int_{S_Q} (Q-X)^{i} d\mathbb{G}} = {\theta}_{1,i}$ and $E(X-Q)^{i}= {\theta}_{2,i}$, ${\forall}~i = 1,2,\ldots$,  Eq.~\ref{FOC} can be written as
\begin{equation}
h(\vec{\theta},Q)={\frac{{\theta}_{1,m-1}}{{\theta}_{2,m-1}}} = k_m
\label{finalFOC}
\end{equation}
Let us define the $j^{th}$ partial raw moment of $X$ as $\delta_j=\int_{S_Q} X^j d\mathbb{G}$ and the $j^{th}$ raw moment of $X$ by ${\mu}_j^\prime=\int_\mathcal{X} X^j d\mathbb{G}$ $\forall j=1,2,\ldots $. Further let, the optimal expected cost be denoted by $\varphi_m^*$ and the corresponding set of optimal order quantities by $\mathcal{U^*}$, which are obtained by solving the population stochastic minimisation problem in  Eq.~\ref{sygennv}.
Next we show that $\mathcal{U^*}$ is non-empty, \emph{i.e.} at least one feasible solution to Eq.~\ref{finalFOC} exists. 
% First we state the following theorem without proof giving existence of a real root of a continuous function.

% \begin{theorem}[\textbf{Bolzano's Theorem \cite{apostol1991calculus}}]\label{bolzanotheorem}
% Let f be continuous at each point of a closed interval [a, b] and assume that f(a) and f(b) have opposite signs. Then there is at least one c in the open interval (a, b) such that f (c) = 0
% \end{theorem}
\begin{theorem}\label{rootexist}
 Consider the stochastic minimisation problem in a SyGen-NV set-up as follows,
\begin{equation}
    \underset{Q\in \mathcal{X}}{argmin}~ E_G[C_m(Q,X)]
    \label{sygennv}
\end{equation}
where $X$ is the positive demand defined over the probability space $(\Omega, \mathcal{F},\mathbb{P})$ and $Q$ is order quantity. Under the assumptions A1-A3 and 
\begin{itemize}
     \item[I.] if $m$ is even, then there will exist at least one positive solution to the stochastic minimisation problem provided $\beta_j = \binom{m-1}{j} [\delta_j- (-1)^{m-1}k_m \mu_j^\prime]$ are of the same sign for at least two consecutive $j$'s $(j \in \{0,1,\ldots ,m-1\})$.
     \item[II.] if $m$ is odd, then at least one positive solution to the stochastic minimisation problem will exist. 
\end{itemize}
% Or $T_j-(-1)^{m-1}k_m {\mu}_j \leq 0.$ $\forall j= 1(1)m-1$ at least for two consecutive j's. 
\end{theorem}
\begin{proof}
From the first order condition in Eq.~\ref{finalFOC}, we notice that 
% \begin{equation}
%     \sum_{j=0}^{m-1} \binom{m-1} {j} Q^{m-1-j} (-1)^{j} [\delta_{j}-(-1)^{m-1}k_m{\mu}_{j}^\prime] = 0
%     \label{altfinalFOC}
% \end{equation}
% where $\delta_j$ and ${\mu}_j^\prime$ are as defined above. 

% From equation $(\ref{finalFOC})$ we have 
\begin{eqnarray} 
    & &\int_{S_Q} (Q-X)^{m-1} d\mathbb{G} =  k_m (-1)^{m-1} \int_\mathcal{X} (Q-X)^{m-1} d\mathbb{G}, ~(Q\in \mathcal{X}) \nonumber \\
    &\Rightarrow & \int_{S_Q} \sum_{j=0}^{m-1} \binom{m-1} {j} Q^{m-1-j} (-1)^{j} X^j d\mathbb{G}  =  k_m (-1)^{m-1} \int_\mathcal{X} (Q-X)^{m-1} d\mathbb{G} \nonumber \\
    &\Rightarrow & \sum_{j=0}^{m-1} \binom{m-1} {j} Q^{m-1-j} (-1)^{j} \left[  \int_{S_Q} X^j d\mathbb{G}  - k_m (-1)^{m-1} \int_\mathcal{X} X^j d\mathbb{G} \right] =  0.  \nonumber \\
    &\Rightarrow & \sum_{j=0}^{m-1} \binom{m-1} {j} Q^{m-1-j} (-1)^{j} [\delta_{j}-(-1)^{m-1}k_m{\mu}_{j}^\prime] = 0 \nonumber \\
    &\Rightarrow & \sum_{j=0}^{m-1} (-1)^{j} \beta_j Q^{m-1-j} = 0, \, where \; \beta_j= \binom{m-1} {j}[\delta_{j}-(-1)^{m-1}k_m{\mu}_{j}^\prime]
    \label{altfinalFOC}
\end{eqnarray}
% The order of integration is changed as the order of sum is finite.\\
% In the following theorem we provide a proof for the existence of a real zero of the above equation.\\
\par If m is odd $(m=2d+1)$, then the polynomial is an even degree one. Observe that, in this case $0<k_m< 1$ and  $\beta_j = \binom{2d}{j} [\delta_j - k_{2d+1} \mu_j^\prime]$. Letting $Q \rightarrow 0$, it can be observed that, $\delta_{2d} \rightarrow 0$, resulting in $\displaystyle \lim_{Q\rightarrow 0} \beta_{2d}=-k_{2d+1}\mu_{2d}^\prime <0 $ so that $\displaystyle \lim_{Q\rightarrow 0} \sum_{j=0}^{2d} (-1)^{j} \beta_j Q^{2d-j} = \beta_{2d} < 0$ . 
\par On the other hand, it is possible to choose a large Q, say $Q_0$, so that $\delta_j \approx \mu_j^\prime,\; \forall j=0,1,\ldots 2d$, whenever $Q\geq Q_0$. In that case, $\beta_j \rightarrow \tau_j$, where, $\tau_j= \binom{2d}{j} \mu_j^\prime (1-k_{2d+1}) > 0,\; \forall j=0,1,\ldots 2d$. Choosing $\displaystyle Q_0=\max \left \{ \frac{\tau_{2j+1}}{\tau_{2j}}: j=0,1,\ldots d\right\}$, we, therefore, obtain
\begin{eqnarray}
    \sum_{j=0}^{2d} (-1)^j \tau_j Q^{2d-j}& = & \tau_0 Q^{2d}-\tau_1 Q^{2d-1}+\ldots+\tau_{2d-2} Q^2-\tau_{2d-1} Q +\tau_{2d} \nonumber\\
    & = & Q^{2d-1}(\tau_0 Q-\tau_1)+Q^{2d-3}(\tau_2 Q-\tau_3)+\ldots \nonumber \\
    & & +Q(\tau_{2d-2} Q-\tau_{2d-1})+ \tau_{2d} \nonumber \\
    & > & 0, \mbox{ for } Q>Q_0\nonumber
\end{eqnarray}
Thus, the polynomial in Eq.~\ref{altfinalFOC} is negative when $Q\rightarrow 0$ and is positive for large $Q$ (\emph{i.e.} $Q>Q_0$). Hence, presence of a positive solution of Eq.~\ref{altfinalFOC} follows from the well known Bolzano's theorem on zero of continuous functions.

\par If $m$ is even, then the polynomial in the left hand side of Eq.~\ref{altfinalFOC} is an odd degree polynomial. 
% $k_m >0 $ provided $C_e>C_s$. In that case, $\beta_j>0,\; \forall~j=0, 1 \ldots ~m-1$ 
Hence, there would exist at least one real solution to the equation from Descarte's sign rule. In this case, $\beta_j =\binom{m-1}{j} \left[ \delta_j+k_m \mu_j^\prime\right]$. Further, if $k_m > 0$, then $\beta_j>0,\; \forall j$, which leads to $m-1$ sign changes in the consecutive terms of the polynomial. Thus, there would be at least one feasible solution to the stochastic minimisation problem  (\emph{i.e.} positive root to the polynomial). If $k_m < 0$, then either $\beta_j > 0$ or $\beta_j < 0$, for each $j=1,2,\ldots m-1$. If $\beta_j$'s are of same sign $\forall j$, then by the previous argument there will be at least one positive root of the polynomial. In this case, replacing $Q$ by $-Q$ in the above polynomial, no sign change would occur between consecutive terms. Hence, the real roots would all be positive. On the other hand, if all the $\beta_j$'s are not of same sign, then it is required that at least one $j ~(\in \{0,1,\ldots, m-2\})$ exists such that $\beta_j$ and $\beta_{j+1}$ are of the same sign, so that there would exist a positive root of the polynomial.   
%then $\beta_j=\binom{m-1}{j} \left[\delta_j-|k_m| \mu_j^\prime\right]$. which
%Arguing in the same manner as in the case of $m=2d+1$, it can be shown, that there would be at least one positive root of the polynomial and hence one feasible solution to the stochastic optimisation problem. 
% For $\beta_j < 0$, the polynomial can be rewritten as $\displaystyle \sum_{j=0}^{m-1} (-1)^{j+1} |\beta_j| Q^{m-1-j}$. Thus in this case also $m-1$ sign changes will be observed, leading to existence of at least one positive root.  
% If any two consecutive $\beta_j$s are of same sign,
% would imply a change of sign in each pair of consecutive terms of the left hand side expression of Eq.~\ref{altfinalFOC}. Using Descartes' rule of sign, there would exist at least one positive root of the equation. 
% On the other hand, if $C_e<C_s$, then $\beta_j=\delta_j-\mid k_m \mid \mu_j^\prime$ would be positive if $C_e >\left(\frac{\mu_j^\prime}{\delta_j}+1\right)C_s$, which ensures existence of positive root.
Since there could be many positive roots, we select the one with maximum magnitude. 
% then the optimal order quantity is chosen as that value of Q for which the second partial derivative $($w.r.t Q$)$ of the Eq.~\ref{expectedcost}) is maximum. 
% This is due to the fact that since Q is large we have $Q>\frac{\tau_{2j+1}}{\tau_{2j}}, \forall j=1,...,d-1$ and $\tau_{2d}>0$.
%On the other hand, if $m$ is odd, then a positive root will exist if 
%Replacing Q by $(-Q)$ in the left hand side (LHS) expression of Eq.~\ref{altfinalFOC}), sign of the terms ( whose sign changes %in the previous case ) remains unchanged. This ensures that there will be no negative root of the given equation.
\end{proof}
% If all the real roots of the polynomial in the above theorem turns out to be negative, then we consider the trivial solution $Q=0$.
% The following theorem describes error obtained in approximating a continuous and differentiable function by a polynomial of degree m.
%\begin{theorem}[\textbf{Taylor's Theorem \cite{davidson2009real}}]\label{taylortheorem}
%Let $f(x)$ belongs to $\mathcal{C}^m[A,B]$ i.e. (f has m continuous derivative) and furthermore assume that $f^{(m+1)}$ is defined and $|f^{(m+1)}(x)| \leq N ~ \forall x \in [A,B]$. Let $a \in [A,B]$ and $P_m(x)$ be the approximating Taylor polynomial of order $m$ for $f$ at $a$. Then for each $x \in [A,B]$, the error of approximation $R_m(x)=f(x)-P_m(x)$ satisfies
%\begin{equation}
%|R_m(x)| \leq \frac{N (x-a)^{m+1}}{m+1!}
%\label{taythm}
%\end{equation}
%\end{theorem}

\section{Non-parametric optimal order quantity estimation in SyGen-NV}
In this section, we present non-parametric estimation of the optimal order quantity, when the demand distribution is completely unknown, but historical uncensored demand data are available. Let us denote an uncensored random sample of size $n$  by $\vec{X}=(X_1,X_2,...,X_n)^\prime$ drawn from $\mathbb{G}$. We define two statistics $T_{in}(\vec{X}):\mathbb{R^+}^n\rightarrow \mathbb{R^+}, ~(i=1,2)$ as $\displaystyle T_{1n}=\frac{1}{n}\sum_{i=1}^n (Q-X_i)^{m-1}\mathbb{I}(X_i\leq Q)$ and $T_{2n}=\displaystyle \frac{1}{n}\sum_{i=1}^n (X_i-Q)^{m-1}$. % We first discuss some useful properties of the estimating function $h(\vec{T}_n;Q)$ in Eq.~\ref{esteqn}.
Then the sample version of the first order condition in Eq.~\ref{sygennv} can be constructed by replacing $\theta_{i,m-1}$ with corresponding $T_{in}$, ~i=1,2. The estimating equation can be written as  

\begin{align}
       & \label{esteqn} h(\vec{T}_n;Q)   =  \frac{T_{1n}}{T_{2n}}=k_m
\end{align}
Further, we define sample partial and complete raw moments of order $j$ as $d_j\displaystyle =\frac{1}{n}\sum_{i=1}^n X_i^jI(X_i\leq Q)$ and $\displaystyle m_j^\prime=\frac{1}{n}\sum_{i=1}^n X_i^j$. It can be easily observed that the sample raw moments $d_j$ and $m_j^\prime$ are unbiased estimators of $\delta_j$ and $\mu_j^\prime$. Hence, $\hat{\beta}_j=\binom{m-1}{j}[d_j-(-1)^{m-1} k_m m_j^\prime]$ is the unbiased estimator of $\beta_j$. We then construct the sample version of the first order condition provided in Eq.~\ref{altfinalFOC} as  
\begin{align}
    \label{altesteqn} & \sum_{j=0}^{m-1} (-1)^j\hat{\beta}_j Q^{m-1-j} = 0
\end{align}
where $\hat{\beta}_j$ is as defined above. We would refer to $h(\vec{T}_n;Q)$ as estimating function and the polynomial in the alternative form of the first order condition in Eq.~\ref{altesteqn} as the random polynomial estimating function or simply random polynomial.

\subsection{Properties of $\underset{\sim}{T_n}$}
Some important properties of $T_{in},~i=1,2$ are as follows.
\begin{enumerate}
    \item[P1.] $T_{i,n}$ is unbiased for $\theta_{i,m-1}$, $i=1,2$.
    \item[P2.] $T_{i,n}\stackrel{a.s.}{\rightarrow}\theta_{i,m-1}$ as $n\rightarrow \infty$
    \item[P3.] $\sqrt{n}(T_{in}-\theta_{i,m-1})\stackrel{\mathcal{L}}{\rightarrow}N(0,\sigma_{i,n}^2)$, where $n\sigma_{i,n}^2=\theta_{i,2m-2}-\theta_{i,m-1}^2, ~i=1,2$ and the symbol $\stackrel{\mathcal{L}}{\rightarrow}$ stands for convergence in distribution.
\end{enumerate}
Proof of P1 is immediate by taking expectation of $T_{i,n}$. P2 follows from Kolmogorov's strong law of large number \cite[see pp-115][]{rao1973linear} and the fact that each of $T_{i,n}, ~i=1,2$ is an average of independently and identically distributed (iid) random variables satisfying existence of variance by assumption A3 stated above. P3 is also straight forward from Lindeberg-Levy central limit theorem for iid samples \cite{rao1973linear}.

% \ref{esteqn}   \ref{altesteqn}

% Despite the fact that $T_{1n}$ and $T_{2n}$ are sample average approximation (SAA) \cite[see][]{shapiro2014lectures} of $\theta_{1n}$ and $\theta_{2n}$, the estimating equation (eq.~\ref{esteqn}) is not so.

\subsection{Properties of $h ( \underset{\sim}{T_n} ; Q)$}
We begin with the statement of the following properties of $h(\vec{T}_n;Q)$. 
\begin{enumerate}
    \item[P4] $h(\vec{T}_n;Q)$ is a measurable function over $({\mathbb{R}^+}^n,\mathcal{B}_n)$ for every $Q\in \mathcal{X}$.
    \item[P5] $h(\vec{T}_n;Q)$ is continuously differentiable with respect to $Q$ within the compact set $\mathcal{X}$ a.e $\mathcal{B}_n$.
\end{enumerate}
Property P4 of $h(\vec{T}_n;Q)$ is straight forward from the fact that it is a ratio of two measurable functions (\emph{viz.} polynomials) for every $Q\in \mathcal{X}$. The next property follows from the facts that $T_{1n}$ and $T_{2n}$ are positive $a.e~{\mathbb{R}^+}^n$ for every $Q\in \mathcal{X}$ and ratio of non-zero polynomials are differentiable. 

In what follows, we provide the asymptotic distribution of the random function $h(\vec{T}_n;Q)$ for every $Q\in \mathcal{X}$. First we state an important result, called the delta method for asymptotic normality of a one time differentiable function.
\begin{theorem}[\textbf{Delta Method \cite{dasgupta2008asymptotic}}]\label{deltamethod}

    Suppose ${\vec{W}_n}$ is a sequence of $k$-dimensional random vectors such that $\sqrt{n}(\vec{W}_n-\vec{\theta}) ~\stackrel{\mathcal{L}}{\rightarrow} ~ N_k(\vec{0},\Sigma)$. Let $g : \mathbb{R}_k\rightarrow \mathbb{R}$ be once differentiable at $\theta$ with the gradient vector $g^{(1)}(\theta)$. Then
    \begin{equation}
        \sqrt{n}(g(\vec{W}_n)-g(\vec{\theta})) \stackrel{\mathcal{L}}{\rightarrow} N(0, {g^{(1)}}^\prime(\theta)\Sigma g^{(1)}(\theta) )
    \end{equation}
\end{theorem}
We now prove the asymptotic normality of $h(\vec{T}_n;Q)$ in the following theorem.
\begin{theorem} 
Consider the estimating function $h(\vec{T}_n;Q)$ in Eq.~\ref{esteqn}. Then for large $n$
\begin{equation}
    \sqrt{n}(h(\vec{T}_{n};Q)-h(\vec{\theta};Q))~ {\overset{\mathcal{L}}{\to}}~ N\left(0,{\vec{h}^{(1)}}^\prime ~\Sigma~\vec{h}^{(1)} \right) \;
\end{equation}
where $\Sigma$ is the dispersion matrix of $\vec{T}_n$, $\vec{h}^{(1)}$ is the $1^{st}$ vector derivative of $h(\vec{T}_n;Q)$ with respect to $\vec{T}_n$ evaluated at $\vec{\theta}$ and  \[{\vec{h}^{(1)}}^\prime ~\Sigma~\vec{h}^{(1)}= h(\vec{\theta};Q)^2\left[ {\frac{\theta_{1,2m-2}}{\theta_{1,m-1}^2}}+\frac{\theta_{2,2m-2}}{\theta_{2,m-1}^2}+2(-1)^{m}\frac{\theta_{1,2m-2}}{\theta_{1,m-1}\theta_{2,m-1}}\right]\]
\end{theorem}
\begin{proof}
The co-variance between $T_{1n}$ and $T_{2n}$ is
\begin{eqnarray}
    \sigma_{12;n}&=&Cov({T_{1n},T_{2n}}) \nonumber \\
    &=&  Cov\left({\frac{1}{n}}{\sum_{i=1}^{n}}(Q-X_i)^{m-1}I(X_i \leq Q),{\frac{1}{n}}{\sum_{i=1}^{n}}(X_i-Q)^{m-1}\right) \nonumber \\
    & = & {\frac{1}{n^2}}{\sum_{i=1}^{n}}Cov((Q-X_i)^{m-1}I(X_i \leq Q),(X_i-Q)^{m-1}) \nonumber \\
    &=& {\frac{1}{n}}\left[(-1)^{m-1}{\theta}_{1,2m-2}-{\theta}_{1,m-1}{\theta}_{2,m-1}\right]
    \label{covT1T2}
\end{eqnarray}
From the property P3 and Eq.~\ref{covT1T2}, it could be easily seen that $\sqrt{n}\left(\vec{T}_n-\vec{\theta}\right)$ is asymptotically multivariate normal with dispersion matrix $\Sigma=((\sigma_{ij;n})), \; i,j=1,2$ and $\sigma_{ii;n}=\sigma_{i,n}^2$. Also, note that $T_{in}>0\; a.e. \; {\mathbb{R}^+}^n,\; i=1,2$ and $h(\vec{T}_{n};Q)$ is once differentiable for every $Q\in\mathcal{X}$. We denote the $1^{st}$ derivative of $\vec{h}(\vec{T}_n;Q)$ by $\vec{h}^{(1)}=(h^1(\vec{\theta};Q), h^2(\vec{\theta};Q))^\prime =\left(\frac{1}{\theta_{2,m-1}},\; -\frac{\theta_{1,m-1}}{\theta_{2,m-1}^2} \right)^\prime$, where $h^i(\vec{\theta};Q)=\left. \frac{\partial h(\vec{T}_n;Q)}{\partial T_{in}}\right|_{\vec{T}_n=\vec{\theta}}$ for $i=1,2$. Thus, using routine algebra it can be easily shown that 
\begin{eqnarray*}
{\vec{h}^{(1)}}^\prime~\Sigma~{\vec{h}^{(1)}} & = & h(\vec{\theta};Q)^2 \left[{\frac{\sigma_{1n}^2}{{\theta}_{1,m-1}^2}}+{\frac{\sigma_{2n}^2}{{\theta}_{2,m-1}^2}}-2{\frac{\sigma_{12,n}}{\theta_{1,{m-1}}\theta_{2,m-1}}}\right] \\
& = & \frac{h(\vec{\theta};Q)^2}{n} \left[\frac{\theta_{1,2m-2}}{\theta_{1,m-1}^2}+\frac{\theta_{2,2m-2}}{\theta_{2,m-1}^2}+2(-1)^m \frac{\theta_{1,2m-2}}{\theta_{1,m-1}\theta_{2,m-1}} \right]
\end{eqnarray*}
The proof of the theorem is then immediate from the delta method (Th. \ref{deltamethod}).

\end{proof}
% Next, we present the asymptotic properties of $T_{1n}$ and $T_{2n}$
% \begin{enumerate}
%     \item $T_{in}\stackrel{a.s}{\rightarrow}\theta_{i,m-1},\; as \; n\rightarrow \infty \;$ for $i=1,2$
    
% \end{enumerate}

\subsection{Solution of the estimating equation}

 In this section we present the statistical properties of the estimated optimal order quantity and the optimal value function. We denote by $\hat{\varphi}^*_m$ the estimated optimal cost function and the corresponding set of estimated optimal order quantities are denoted by $\hat{\mathcal{U}}^*$. In the following theorem we prove that $\hat{\mathcal{U}}^*$ is non-empty with probability ($wp$) 1, \emph{i.e} there exists at least one positive solution to Eq.~\ref{esteqn} $wp~ 1$.

\begin{theorem}
Under the regularity assumptions $A1 -A3$, the random polynomial $\displaystyle \sum_{j=0}^{m-1} (-1)^j\hat{\beta}_j Q^{m-1-j}$ will have positive zeroes $wp$ 1 in the following cases.
\begin{itemize}
     \item[I.] For even $m$, if at least two consecutive $\hat{\beta}_j$'s $(j \in \{0,1,\ldots ,m-1\})$ are of the same sign \emph{wp} 1, then at least one positive solution will exist.
     \item[II.] For odd m, at least one positive solution exists \emph{wp} 1. 
\end{itemize}
% for odd $m$. For even $m$, the zeroes are all positive if $C_e < C_s$. If $C_e>C_s$, then positive zeroes of the random polynomial exists if
where $\hat{\beta}_j=d_j- (-1)^{m-1} k_m{m_j}^\prime,\; \forall~j=1,2 \ldots m-1$.
% are of the same sign for any two  consecutive $j$ ${a.e.}$ ${\mathbb{R}^+}^n$.
\end{theorem}
\begin{proof}
 Notice that, $d_j~ \overset{a.s}{\to}~ \delta_j$ and $m_j^\prime~ \overset{a.s}{\to} ~\mu_j^\prime$, which implies in turn that $\hat{\beta}_j~\overset{a.s}{\to} ~\beta_j$. Thus the proof of this theorem is same as that of Th.~\ref{rootexist} in almost sure sense. We omit the details to avoid repetition. 
 \end{proof}

 Next we show that any solution to the estimating equation converges to the true optimal order quantity in SyGen-NV problem. Let the solution of the estimating equation Eq.~\ref{esteqn} (or Eq.~\ref{altesteqn}) be denoted by $\hat{Q}_n^*$. We show that the solution is strongly consistent for the solution to the stochastic optimisation problem $\underset{Q\in\mathcal{X}}{argmin}~ E_{\mathbb{G}}\left[ C_m(Q,X)\right]$ under mild regularity conditions. First we state the following theorem without proof on existence of optima of a continuous function on a compact set.

\begin{theorem}[\textbf{Extreme Value Theorem \citep[see][]{stein2010complex}}]\label{extremevalue}
A continuous function on a compact set $\mathcal{X}$ is bounded and attains a maximum and minimum on $\mathcal{X}$.
\end{theorem}
We state the next lemma on the compactness of the complement of an open subset of a compact set. 
\begin{lemma}\label{compactlem}
Let $\mathcal{X}$ be a compact set and $O$ be an open subset of $\mathcal{X}$. Then $\Bar{O}= \mathcal{X} \setminus O $, denoting the complement of $O$ in $\mathcal{X}$, is also a compact set.
\end{lemma}
The proof is a routine exercise in real analysis and hence is omitted. 
% \begin{proof}
% Let $V$ be an open cover of $\mathcal{X} - O$. V $\cup$ $\{O\}$ is an open cover of $\mathcal{X}$. There will be a finite sub-cover of this open cover. Let W be the finite sub-cover i.e. W $\subseteq V \cup \{O\}$. Then $W - \{O\} \subseteq V$ is finite and covers $\mathcal{X} - O$. This proves \Bar{O} is compact.  
% \end{proof}

\begin{theorem}\label{extvth}
Let $\hat{Q}_n^*\in \mathcal{X}$ be the unique solution to the estimating equation $h(\vec{T}_n;Q)=k_m$ and $Q^*$ uniquely solves the stochastic programming problem \[\underset{Q\in \mathcal{X}}{argmin}~E_{\mathbb{G}}\left[C_m(Q,X)\right]\] Then
\begin{equation}
    \hat{Q}_n^*\overset{a.s.}{\to}Q^*
\end{equation}
\end{theorem}
\begin{proof}
Let $O \subseteq \mathcal{X}$ denote an arbitrary open neighbourhood of $Q^*$. From lemma~\ref{compactlem}, the complement of $O$, $\bar{O}=\mathcal{X}\setminus O$ is also a compact set. Notice that the expected cost $E_\mathbb{G} [C_m(X,Q)] ~(=\varphi_m (Q), \; say)$, is a continuous function of Q. Hence, from Theorem~\ref{extremevalue},  the stochastic optimisation problem $\underset{Q}{argmin} ~\varphi_m (Q)$ will have a solution in $\Bar{O}$ with unique minimum value of $\varphi_m (Q)$. 
Let us denote, $ \displaystyle r = \min_{Q \in \Bar{O}} \varphi_m (Q) - \varphi_m (Q^*) > 0$. 

Also, from property P2 of $T_{in},\; (i=1,2)$ and the continuous mapping theorem, it can be easily seen that $h(\vec{T}_n,Q) \overset{a.s.}{\to}h(\vec{\theta},Q), ~\forall ~Q\in \mathcal{X}$. Since $\hat{Q}_n^*\in \mathcal{X}$, there would exist $n_0(\epsilon)$ for every $\epsilon>0$, such that $\mid h(\vec{\theta},\hat{Q}_n^*)-k_m \mid<\epsilon$,  $\forall ~ n\geq n_0(\epsilon)$, $wp$ 1. Therefore $\exists~n>n_0(\epsilon)$ for every $0<\epsilon<\frac{r}{2}$, so that
\begin{equation}
|h(\theta,\hat{Q}_n^*) - h(\theta,Q^*)|<\epsilon,  ~\forall ~n>n_0(\epsilon),\; wp\;1
\label{inqmin}
\end{equation}
This implies $\hat{Q}_n^* \notin \bar{O}$. $O$ being arbitrary, $\hat{Q}_n^* \overset{a.s.}{\to} Q^*$. 
\end{proof}
The roots of the FOC (Eq.~\ref{altfinalFOC}) may not be unique. Let the set of corresponding distinct roots be denoted by ${\mathbf{Q}}^*=\{Q_1^*,Q_2^* \ldots Q_k^*\},\; k=1,2\ldots m-1 $. Similarly, there could be $p~(\geq 1)$ roots of the random polynomial (Eq.~\ref{altesteqn}), say $\hat{\mathbf{Q}}^*=\{\hat{Q}_{1}^*,\hat{Q}_{2}^* \ldots \hat{Q}_{p}^*\}$. In the next two corollaries, we extend Theorem~\ref{extvth} for multiple roots. 
\textcolor{black}{
\begin{corollary}
Let $\hat{\mathbf{Q}}^*$ be the set of distinct roots of the random polynomial (Eq.~\ref{altesteqn}) and $Q^*$ be unique solution to the stochastic minimisation problem (\ref{sygennv}). Then $\hat{Q}_{max}^* \overset{a.s}{\to} Q^*$, where $\hat{Q}_{max}^*=\max\{
\hat{\mathbf{Q}}^*\}$.
\end{corollary}
\begin{proof}
Notice, the maximum of $\hat{\mathbf{Q}^*}$  is unique. Hence, from Th.~\ref{extvth}, the proof is immediate. 
\end{proof}
\begin{corollary}
Let $\hat{Q}_n^*$ be the unique solution to the random polynomial equation Eq.~\ref{altesteqn} and $\hat{\mathbf{Q}}^*$ be the set of distinct solutions to the stochastic minimisation problem (\ref{sygennv}). Then $\hat{Q}^* \overset{a.s}{\to} Q_i^*$; for exactly one $i$; $i=i=1,2,\ldots,k$. 
\end{corollary}
\begin{proof}
Let $O_i$ denote an arbitrary open neighbourhood around $Q_i^*$ selected in such a way that $O_i$'s are disjoint. Then, $O = \cup_{i=1}^{k} O_i$ is also an open set. Implementing the same argument as Theorem~\ref{extvth} we ensure that $\hat{Q}_n^* \in O$. Disjoint property of $O_i$ indicates $\hat{Q}_n^* \in O_i$ for exactly one $i$.
\end{proof}
\begin{corollary}
Let $\hat{\mathbf{Q}}^*$ be the set of distinct solutions to the random polynomial equation Eq.~\ref{altesteqn} and $\mathbf{Q}^*$ is the set of distinct solutions of the FOC Eq.~\ref{altfinalFOC}, then $\hat{{Q}}_{max}^* \overset{a.s}{\to} Q_i^*$; for exactly one $i$; $i=i=1,2,\ldots,k$.
\end{corollary}
\begin{proof}
Proof immediately follows from previous two corollaries.
\end{proof}
}

From the above theorem, it can be easily seen that the estimated optimal cost $\hat{\varphi}_n^*=\varphi_m(\hat{Q}^*)$ almost surely converges to the true optimal cost $\varphi^*_m$, using the continuity of the cost function $\varphi_m(Q)$.

\section{Monte-Carlo Simulation experiments}
In this section we present the results of Monte-Carlo simulation experiments on the non-parametric estimator of the optimal order quantity in SyGen-NV set-up. We consider here two known probability distributions for the demand, \emph{viz.} $Uniform(0,1)$ and $Exp(1)$. The severity index $m$ is assumed to be known ($\in \{2,3,4,5,10\}$). Further, we take the excess-to-shortage cost ratio, $\lambda~(=\frac{C_e}{C_s}) \in \{0.25, 0.45, 0.65, 0.85, 1.05, 1.25, 1.45, 1.65, 1.85 \}$. For each of the $(m, \lambda)$ pairs, we compute numerically the optimal order quantities for both $Uniform$ and $Exponential$ true demands. Further, we conduct  3.15 million Monte-Carlo simulation experiments for each of the demand distributions to understand the small and large sample properties of the non-parametric estimator. In particular, we draw random samples of size $n$ $(=20,50,100,500,1000,5000 $, $10000)$ for each combination of $(\lambda, ~m)$ and estimate the optimal order quantities $\hat{Q}_n^*$ therefrom. We repeat this process for $M$ times ($M=5000$). We study the sampling properties of $\hat{Q}_n^*$ from these $M$ estimates.
\subsection{$Unif(0,1)$ Demand distribution }
The optimal order quantity in the SyGen-NV problem with $Unif(0,1)$ demand is given by \citep{mukhoti2021}
\[Q_n^*=\frac{1}{1+\lambda^{\frac{1}{m}}} \]
$\hat{Q}_n^*$ can be obtained, on the other hand, from the estimating equation~(Eq.~\ref{altesteqn}). However, existence of $\hat{Q}_n^*$ is conditional for even $m$. Let $p_{\lambda,2k}^*$ denote the probability of existence of an estimate of the optimal order quantity when $m=2k$ and $\hat{p}^*_{\lambda,2k}$ be its estimate obtained from the $M$ iterations. The following table reports large sample (n=10000) estimates,  $\hat{p}_{\lambda,2k}^*,~k=1,2,5$ and $\forall~\lambda$ considered in the simulation experiments. 

\begin{table}[ht]
\centering
\begin{tabular}{rrrrrrrrrr}
  \hline
 $\lambda$ & 0.25 & 0.45 & 0.65 & 0.85 & 1.05 & 1.25 & 1.45 & 1.65 & 1.85 \\ 
  \hline
m=2 & 1.00 & 1.00 & 1.00 & 0.50 & 0.49 & 0.73 & 1.00 & 1.00 & 1.00 \\ 
  m=4 & 1.00 & 1.00 & 0.55 & 0.51 & 0.45 & 0.49 & 0.50 & 1.00 & 1.00 \\ 
  m=10 & 1.00 & 0.50 & 0.50 & 0.51 & 0.36 & 0.51 & 0.49 & 0.49 & 0.49 \\ 
   \hline
\end{tabular}
\caption{$\hat{p}_{\lambda,2k}^*$ for Uniform demand} 
\end{table}
It may be observed here that it is least probable to obtain $\hat{Q}_n^*$ when $\lambda$ is close to unity. Also, with increasing severity $(m)$, it becomes more difficult to obtain $\hat{Q}_n^*$ as the probability decreases for a given $\lambda$. The probability distribution of the estimated order quantity is presented in the form of box-plots in Fig.~\ref{UnifDenPlot}. For $\lambda<1$, the probability distributions of $\hat{Q}_n^*$ are stochastically larger with increasing severity levels, the distribution for $m=2$ being centred at the highest value among all others. For $\lambda>1$, the distributions of estimated order quantity for even $m$ are different than those of the odd $m$. Odd severity seems to result in stochastically smaller distribution of $\hat{Q}_n^*$. The variation, on the other hand, seems to decrease with severity for all $\lambda$.

Next we present the performance study of $\hat{Q}_n^*$ using the mean square error (MSE) computed from the $M$ estimates as $\displaystyle MSE=\frac{1}{M}\sum_{i=1}^M (\hat{Q}_{in}^*-Q_n^*)^2$. Figures~\ref{MSE1unif}-\ref{MSE9unif} in the appendix presents the MSE's plotted against sample sizes. It could be seen that for $\lambda<1$, the MSEs converge to $0$ with increasing $n$ for all $m$, with worst performance of $\hat{Q}_n^*$ observed at $m=2$. For $\lambda>1$, however, the convergence is slow in case of even $m$.

\subsection{$Exp(1)$ Demand distribution }

The optimal order quantity in the SyGen-NV problem with $Exp(1)$ demand can be obtained from the random polynomial (Eq.~\ref{altfinalFOC}) by replacing the partial and full raw moments by those for the $Exp(1)$ distribution. The modified equation is given as \citep{mukhoti2021}
\[\sum_{j=0}^{m-1}(-1)^j \left(Q\right)^{m-j-1}\frac{1}{(m-j-1)!} =  e^{-Q} \left[{\frac{C_s}{C_e}}-(-1)^m\right] \]
As described in the uniform case, $\hat{Q}_n^*$ can be obtained from the estimating equation~(Eq.~\ref{altesteqn}). Also, $\hat{p}^*_{\lambda,2k}$ can be obtained from the $M$ iterations in a similar manner as in the $Uniform$ demand case. The following table reports large sample (n=10000) estimates,  $\hat{p}_{\lambda,2k}^*,~k=1,2,5$ and $\forall~\lambda$ considered in the simulation experiments. 

\begin{table}[ht]
\centering
\begin{tabular}{rrrrrrrrrr}
  \hline
 & 0.25 & 0.45 & 0.65 & 0.85 & 1.05 & 1.25 & 1.45 & 1.65 & 1.85 \\ 
  \hline
2 & 1.00 & 1.00 & 1.00 & 0.49 & 0.47 & 0.51 & 1.00 & 1.00 & 1.00 \\ 
  4 & 1.00 & 1.00 & 1.00 & 0.55 & 0.18 & 0.55 & 1.00 & 1.00 & 1.00 \\ 
  10 & 1.00 & 1.00 & 0.90 & 0.35 & 0.10 & 0.48 & 0.78 & 0.98 & 1.00 \\ 
   \hline
\end{tabular}
\caption{Probability of Existence of the Optimal order quantity for Exponential demand)} 
\end{table}

The observations are similar to the $uniform$ demand case. The lowest probability of existence of a solution to the estimating equations occur when $\lambda$ is close to unity. It could be observed that obtaining a feasible $\hat{Q}_n^*$ is more difficult for increasing severity $(m)$, specifically near $\lambda=1$. 

Unlike the uniform demand case,  probability distribution of the estimated optimal order quantity increases stochastically with severity for all $\lambda$ (see Fig.~\ref{ExpDenPlot}). Not only the location, the scale (or variance) of the distribution also increases with $m$. 

In terms of MSE, $\hat{Q}_n^*$ performs well asymptotically as the MSE (vs. $n$) curve (see Fig.~\ref{MSE1Exp}-\ref{MSE9Exp}) decreases to zero with increasing sample size (for all $m$ and $\lambda$), the worst performance being observed for $m=10$. The best estimator, in the MSE sense, is obtained for $m=2$ when $\lambda<1$. However, for $\lambda>1$ performance of $\hat{Q}_n^*$ for $m=2$ worsens in small samples. 

% for  Figures~\ref{MSE1unif}-\ref{MSE9unif} in the appendix presents the MSE's plotted against sample sizes. It could be seen that for $\lambda<1$, the MSEs converge to 0 with increasing $n$ for all $m$. For $\lambda>1$, however, the convergence is slow in case of even $m$. 

\section{Discussion}
In this paper we have discussed non-parametric estimation of the optimal order quantity in case of a general newsvendor problem, where the severity of the losses are much more than merely the quantity lost. Major contributions of this paper are two-fold. First we have constructed a non-parametric estimation method for the optimal order quantity in the SyGen-NV problem with power type shortage and excess. Secondly, we have studied the properties and performances of the estimators of the optimal order quantities. 

Our contribution in the non-parametric estimation of the optimal order quantity starts with formulation of an estimating equation from the first order condition using uncensored demand data. We have presented strong consistency of the estimating function and its asymptotic distribution has been derived. Further, we have presented a random polynomial representation of the estimating equation and established feasibility of the solution by deriving conditions for existence of the zeroes of the random polynomial in almost sure sense. We have also proven the strong consistency of the estimated optimal order quantity.

The theoretical results in this paper has been supported by an exhaustive set of simulation experiments. In particular, we have considered known uniform and exponential as true demand distributions. For each of the demand distributions, we have estimated the probabilities of existence of positive zeroes of the estimating random polynomial. The results show that it is least likely to get an estimate of the optimal order quantity if the cost ratio is close to one. The distribution of the estimated optimal order quantities suggests that odd and even order of severity influences the estimates differently for uniform demand, whereas for exponential demand, the estimate increases uniformly with severity. Comparing the mean square errors for different sample sizes, severity and cost-ratio, it has been found that the estimators perform well in the MSE sense when severity is high in case of uniform demand and the opposite for exponential distribution.

 We conclude the paper with comments on future scope of research. A natural extension of the SyGen-NV problem would be to consider asymmetric weight functions for shortage and excess. Complexity arises due to different dimensions of the two costs as a result of asymmetric weighing. \cite{baraiya2019} discussed, in an unpublished manuscript, selection of weights so that the shortage and excess costs remain comparable. However, estimation of optimal order quantity in such asymmetric generalised newsvendor problem remains open.
 
 \section*{Acknowledgement}
 {This work was supported by the Indian Institute of Management Indore [SEED grant number no. RS/09/2019-20]. The authors would like to thank Dr. Abhirup Banerjee, Institute of Biomedical Engineering, University of Oxford for helpful suggestions on the simulation experiments. 

\bibliography{Reference}

 \newpage
\appendix
\section{} 
\subsection{Figures}\label{Figures}
\begin{figure}[h]
     \centering
     \begin{subfigure}[h]{0.3\textwidth}
         \centering
         \includegraphics[width=\textwidth]{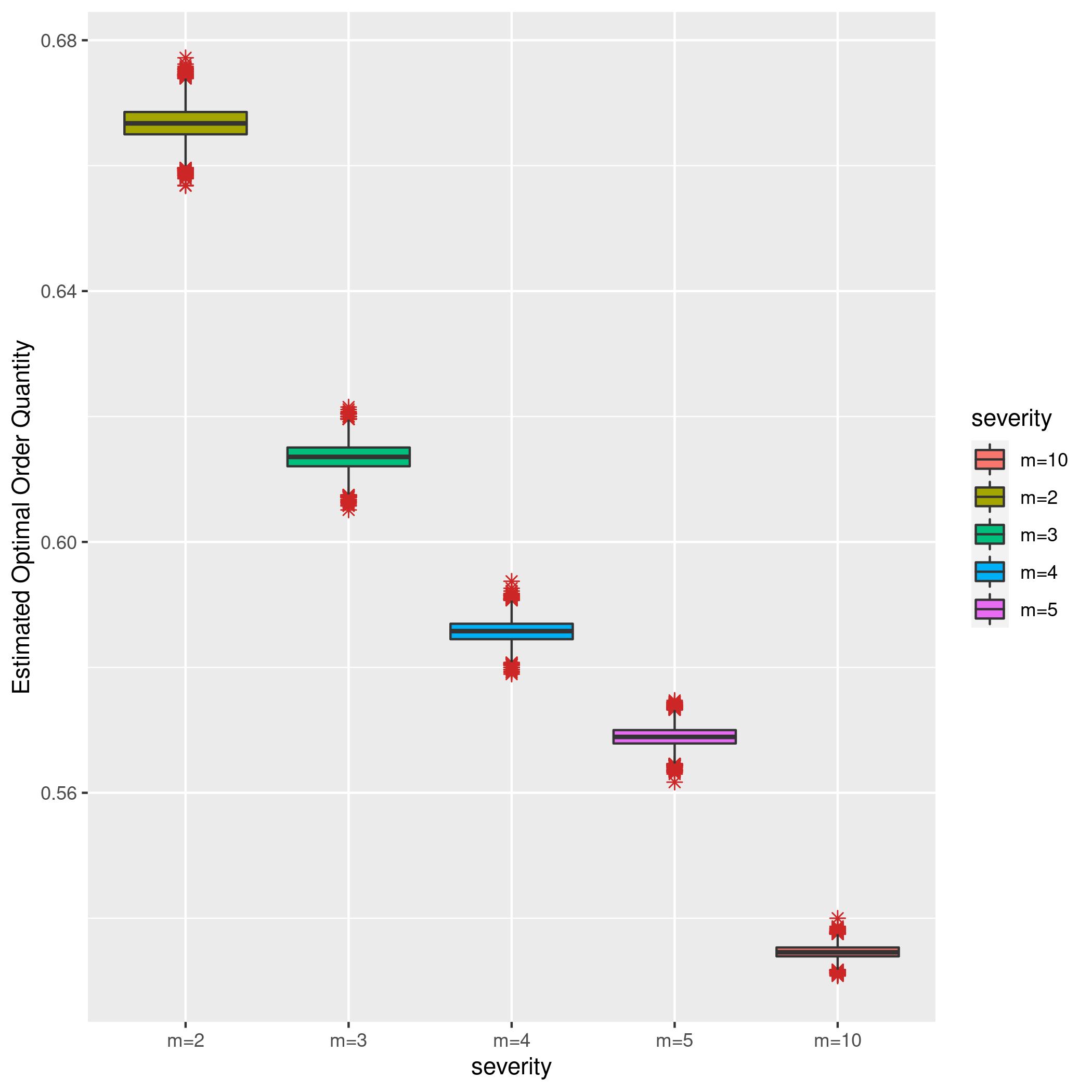}
         \caption{$\lambda=0.25$}
         \label{Qhat1unif}
     \end{subfigure}
     \hfill
     \begin{subfigure}[h]{0.3\textwidth}
         \centering
     \includegraphics[width=\textwidth]{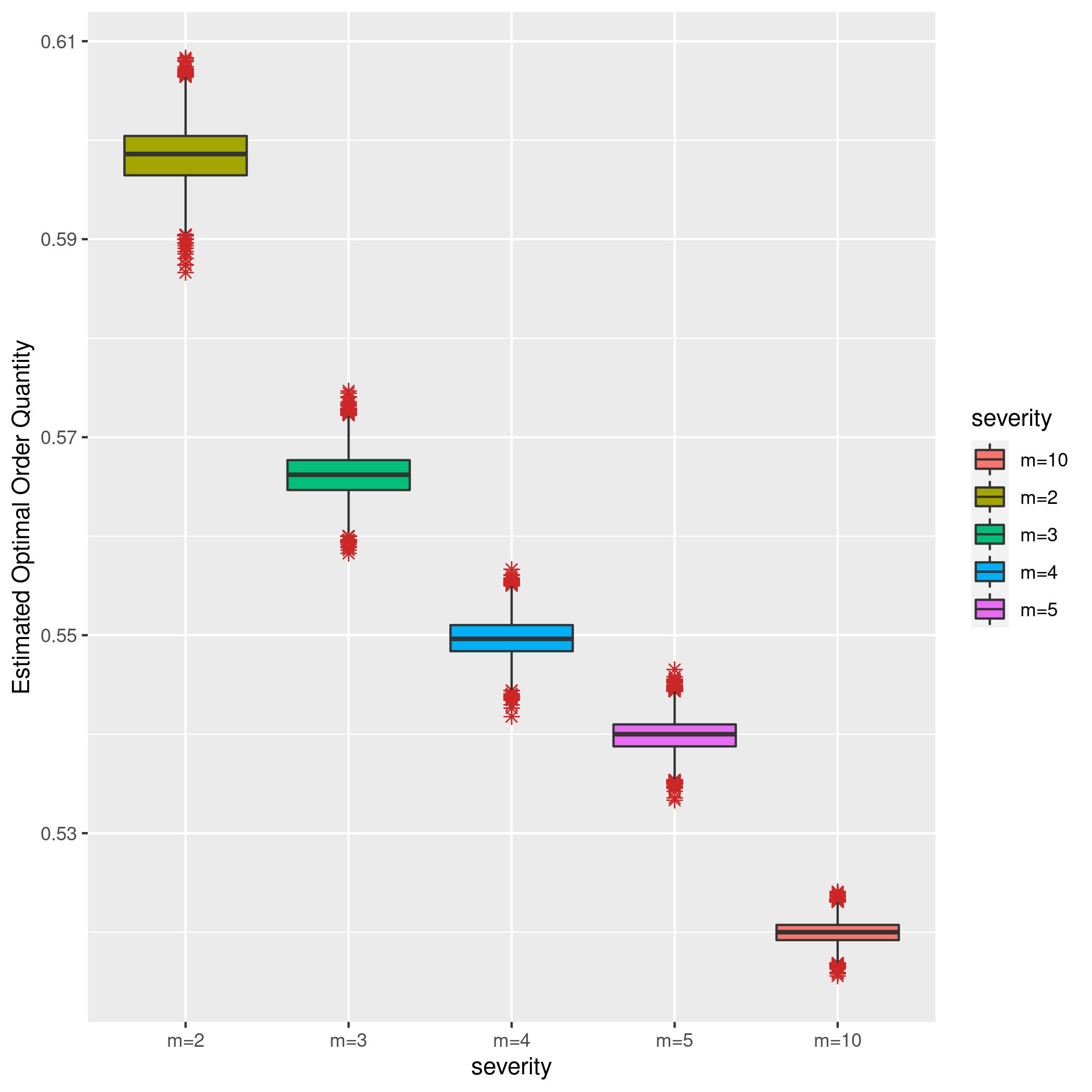}
     \caption{$\lambda=0.45$}
         \label{Qhat2unif}
     \end{subfigure}
     \hfill
     \begin{subfigure}[h]{0.3\textwidth}
         \centering \includegraphics[width=\textwidth]{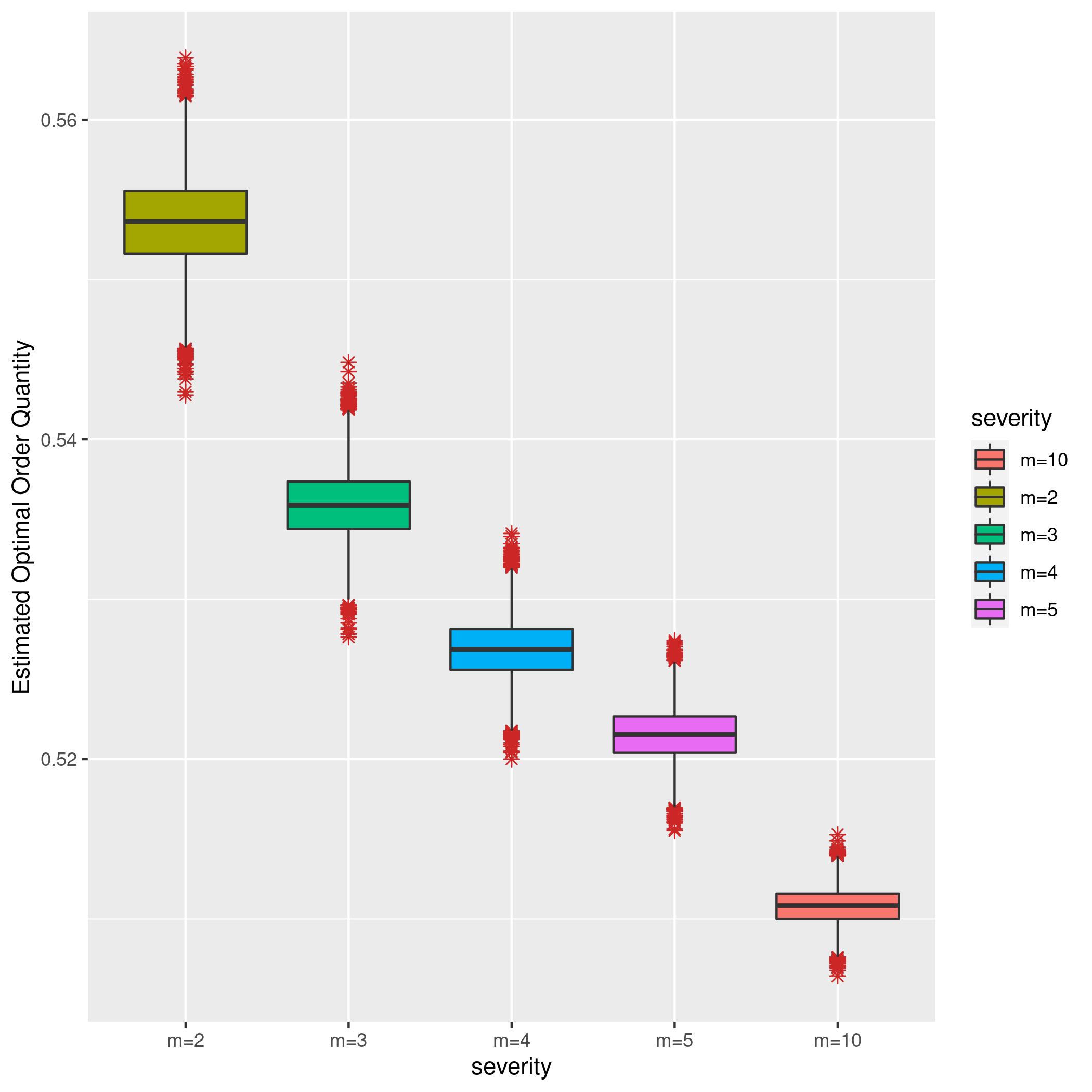}
         \caption{$\lambda=0.65$}
         \label{Qhat3unif}
     \end{subfigure}
     \\
     \begin{subfigure}[h]{0.3\textwidth}
         \centering
         \includegraphics[width=\textwidth]{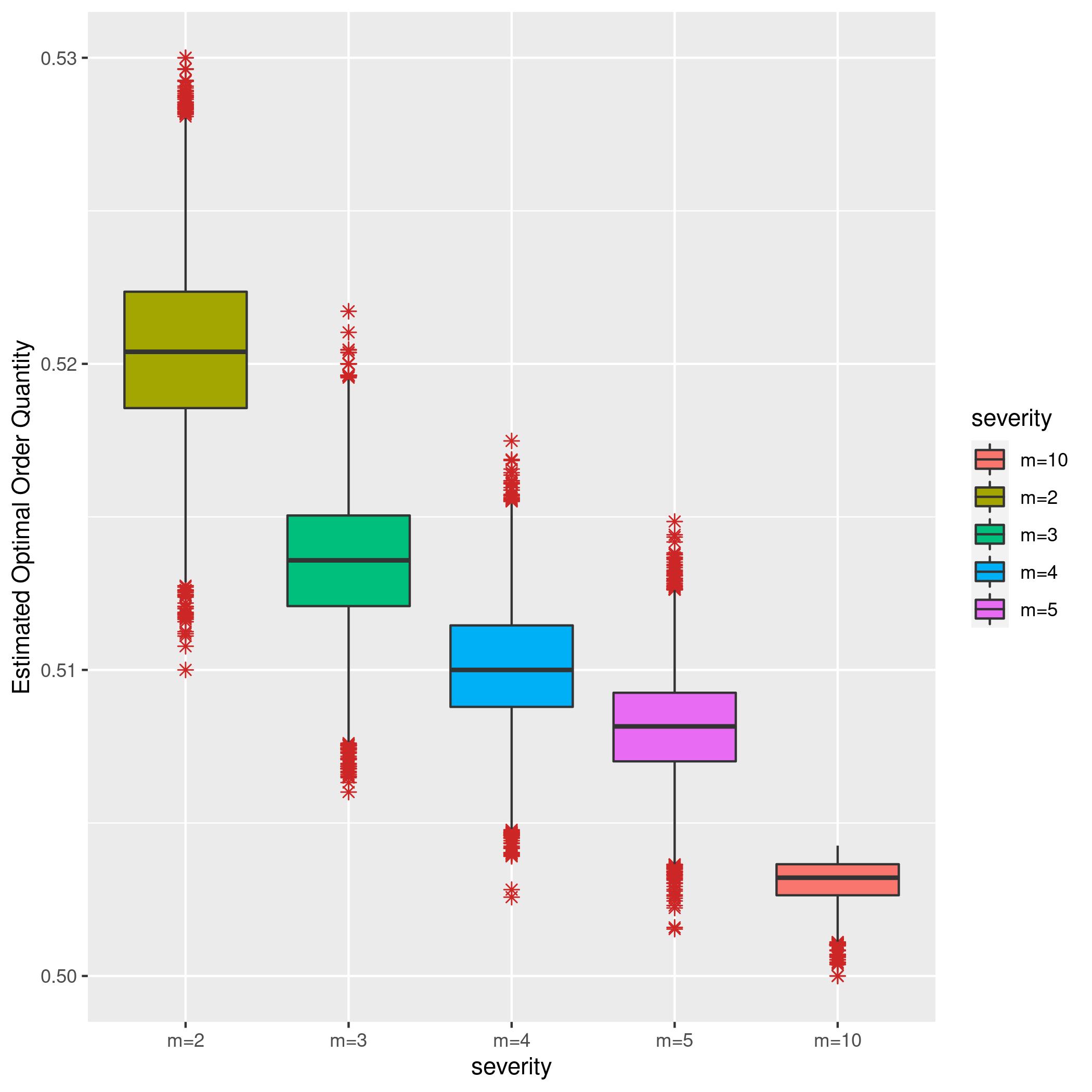}
         \caption{$\lambda=0.85$}
         \label{Qhat4unif}
     \end{subfigure}
     \hfill
     \begin{subfigure}[h]{0.3\textwidth}
         \centering
     \includegraphics[width=\textwidth]{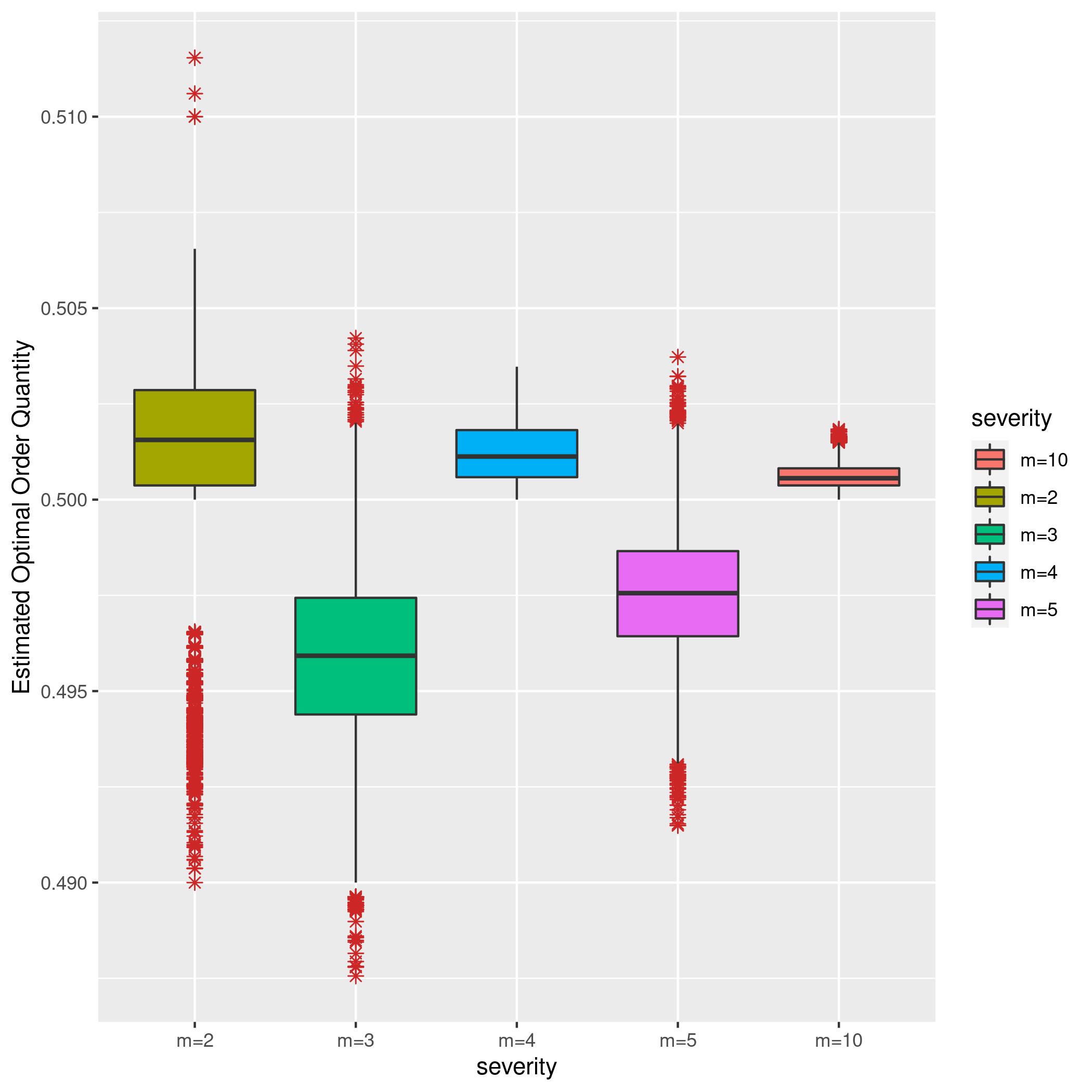}
     \caption{$\lambda=1.05$}
         \label{Qhat5unif}
     \end{subfigure}
     \hfill
     \begin{subfigure}[h]{0.3\textwidth}
         \centering \includegraphics[width=\textwidth]{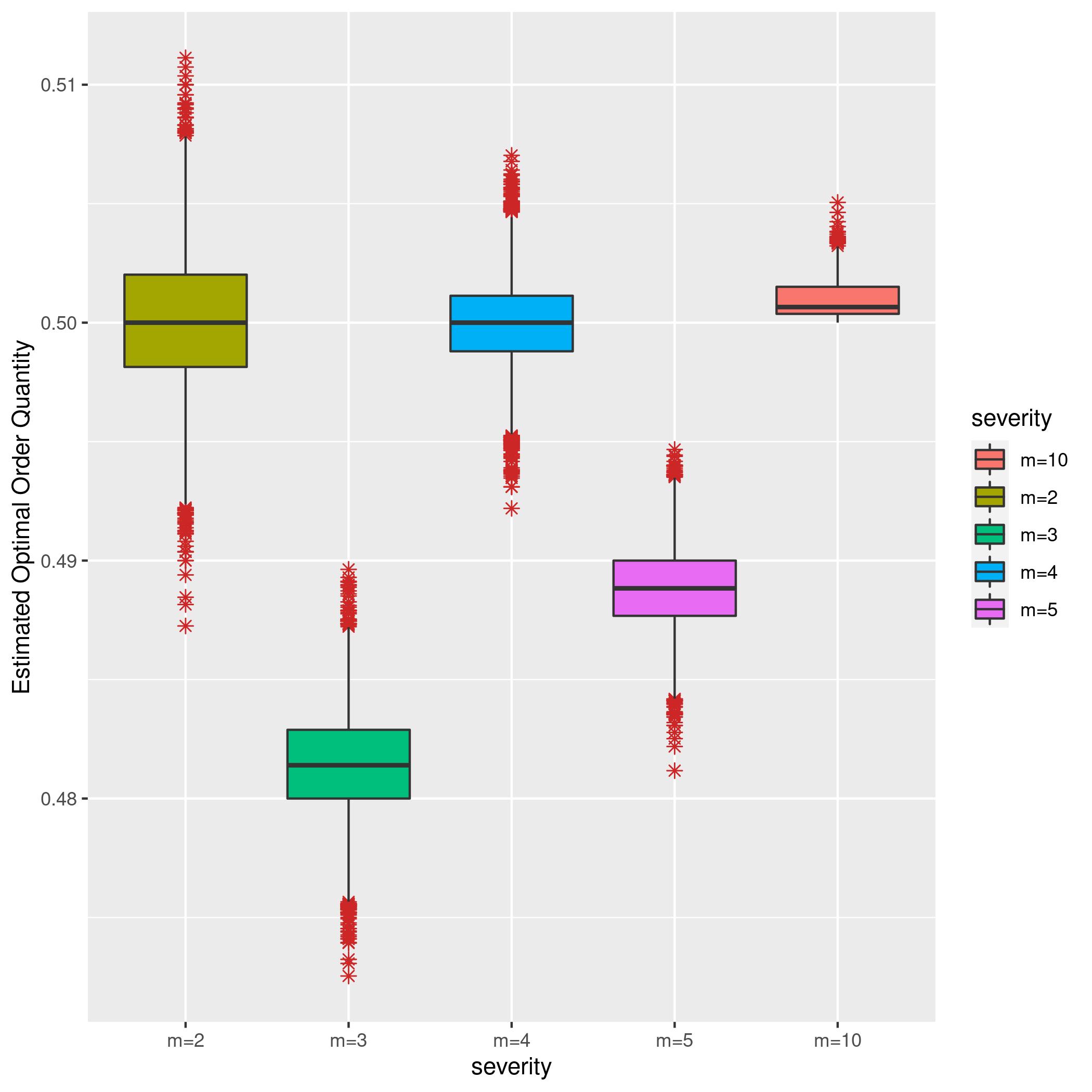}
         \caption{$\lambda=1.25$}
         \label{Qhat6unif}
     \end{subfigure} \\
          \begin{subfigure}[h]{0.3\textwidth}
         \centering
         \includegraphics[width=\textwidth]{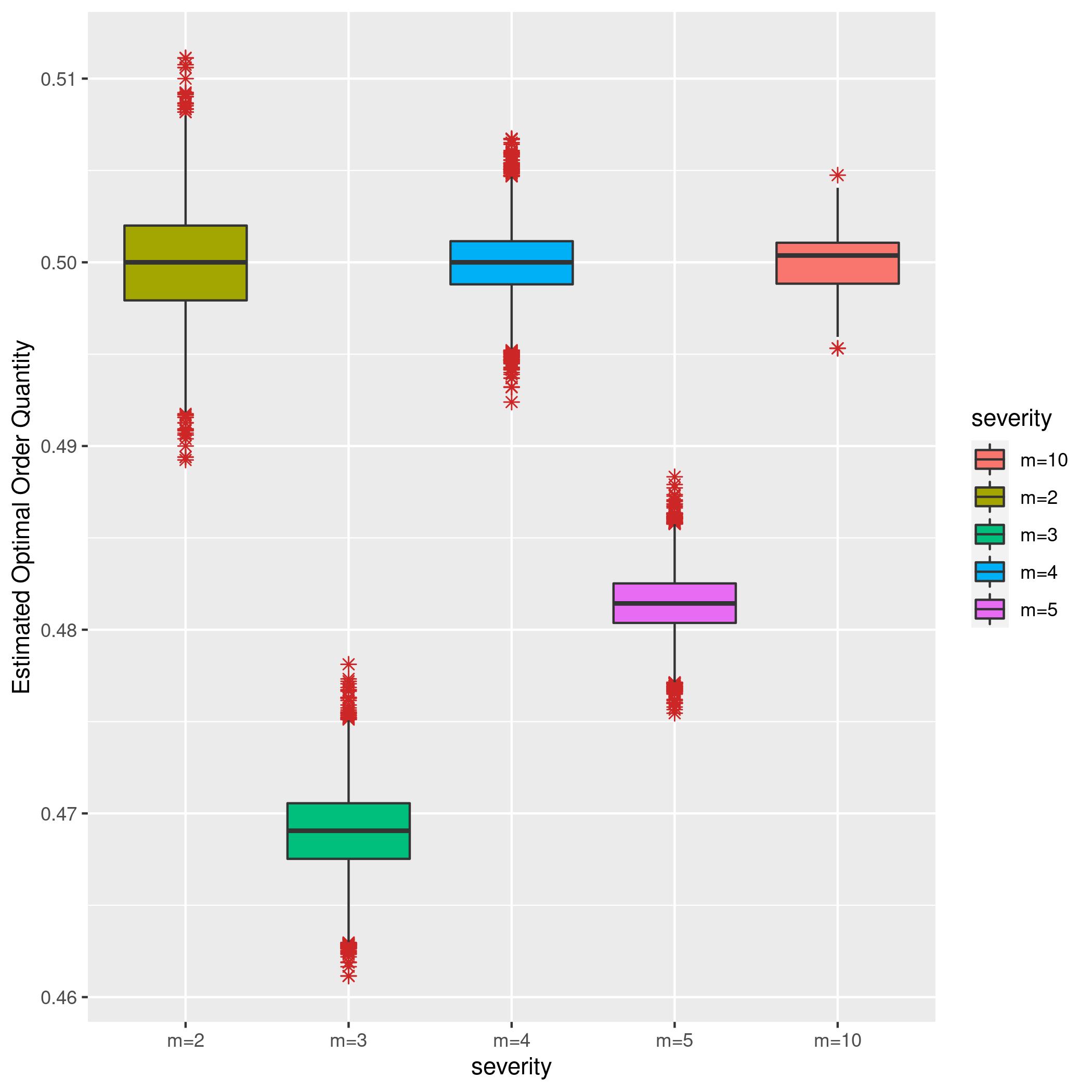}
         \caption{$\lambda=1.45$}
         \label{Qhat7unif}
     \end{subfigure}
     \hfill
     \begin{subfigure}[h]{0.3\textwidth}
         \centering
     \includegraphics[width=\textwidth]{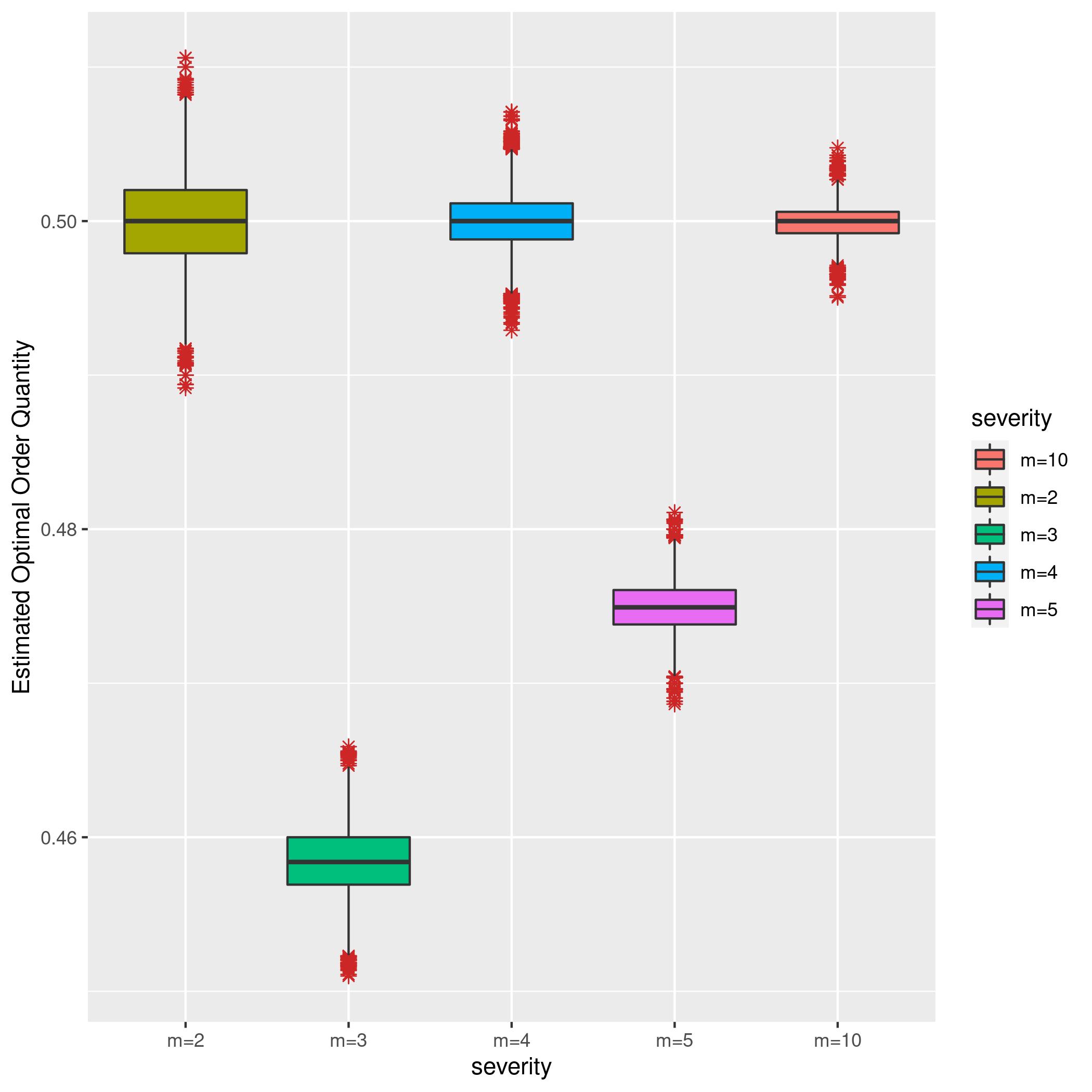}
     \caption{$\lambda=1.65$}
         \label{Qhat8unif}
     \end{subfigure}
     \hfill
     \begin{subfigure}[h]{0.3\textwidth}
         \centering \includegraphics[width=\textwidth]{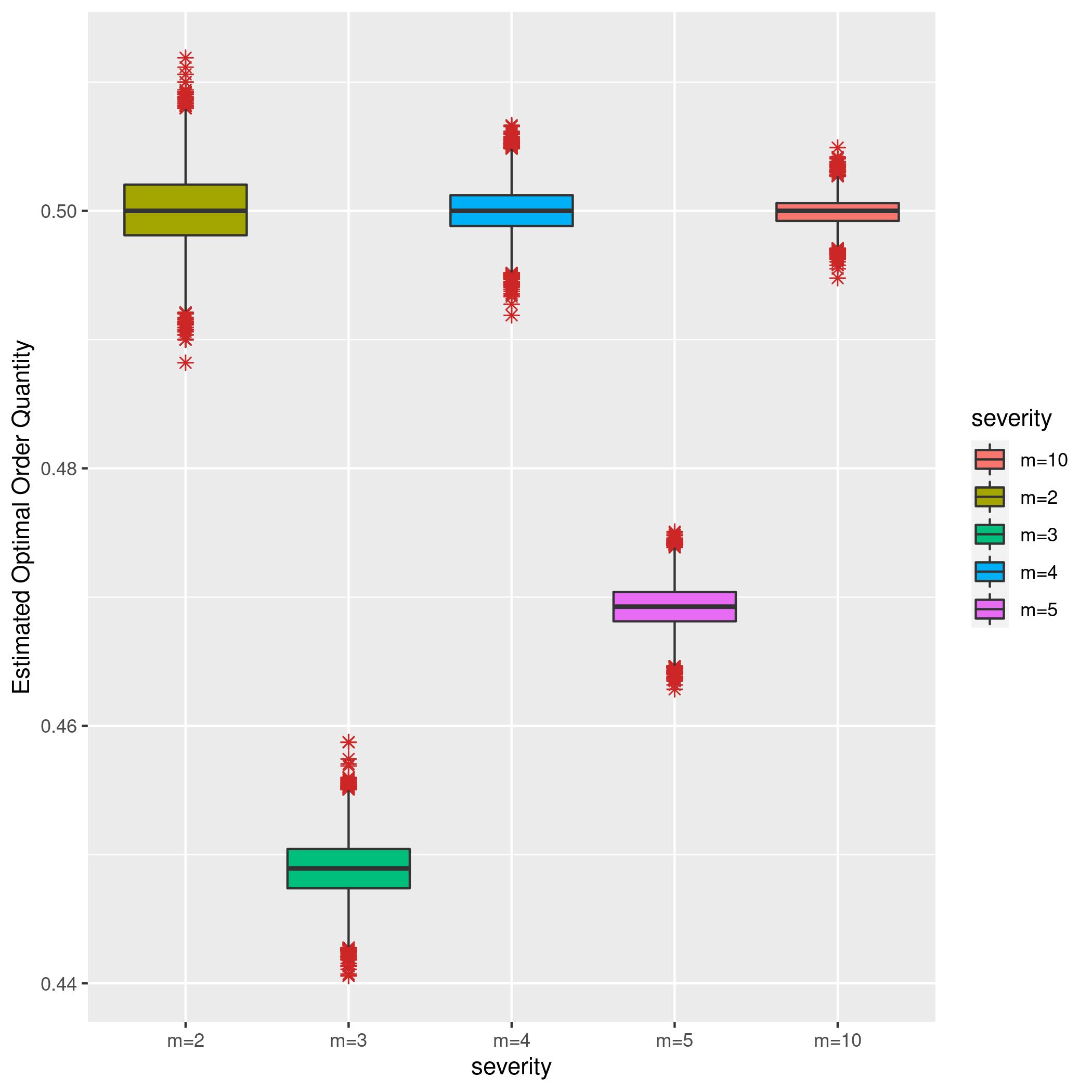}
         \caption{$\lambda=1.85$}
         \label{Qhat9unif}
     \end{subfigure}
        \caption{Boxplot of estimated order quantity for different degrees of severity ($m$) }
    \label{UnifDenPlot}
\end{figure}
\begin{figure}[h]
     \centering
     \begin{subfigure}[h]{0.3\textwidth}
         \centering
         \includegraphics[width=\textwidth]{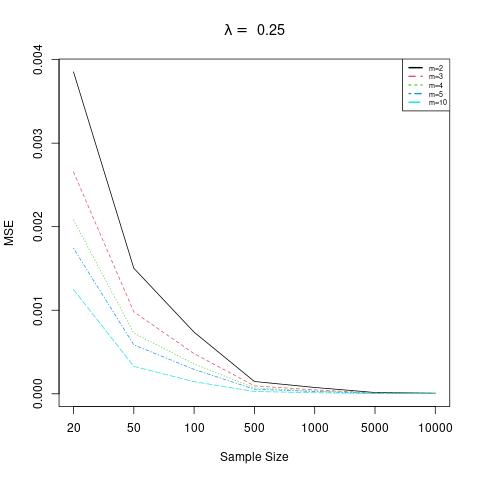}
         \caption{$\lambda=0.25$}
         \label{MSE1unif}
     \end{subfigure}
     \hfill
     \begin{subfigure}[h]{0.3\textwidth}
         \centering
     \includegraphics[width=\textwidth]{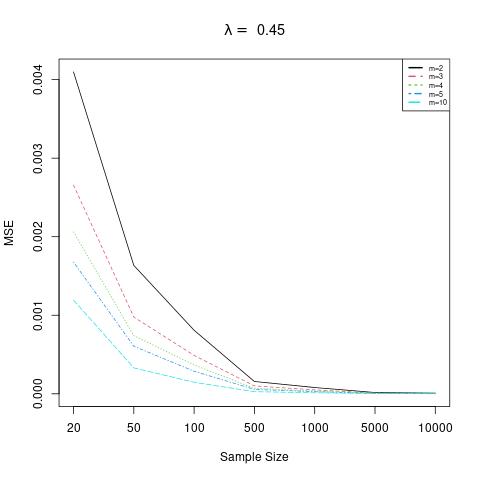}
     \caption{$\lambda=0.45$}
         \label{MSE2unif}
     \end{subfigure}
     \hfill
     \begin{subfigure}[h]{0.3\textwidth}
         \centering \includegraphics[width=\textwidth]{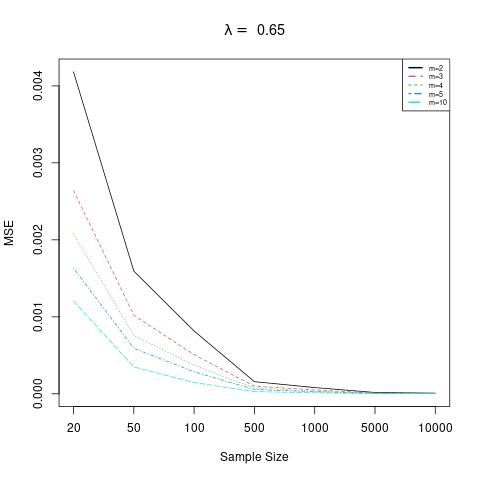}
         \caption{$\lambda=0.65$}
         \label{MSE3unif}
     \end{subfigure}
     \\
     \begin{subfigure}[h]{0.3\textwidth}
         \centering
         \includegraphics[width=\textwidth]{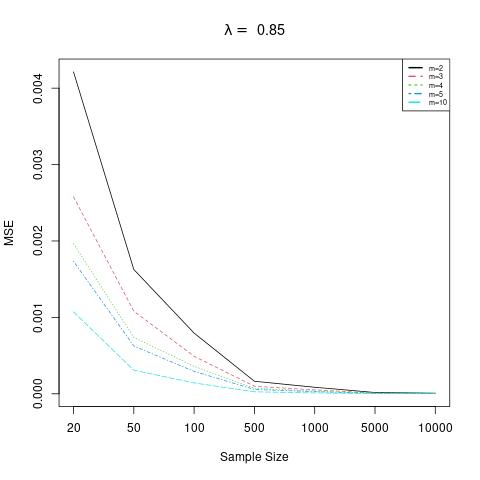}
         \caption{$\lambda=0.85$}
         \label{MSE4Unif}
     \end{subfigure}
     \hfill
     \begin{subfigure}[h]{0.3\textwidth}
         \centering
     \includegraphics[width=\textwidth]{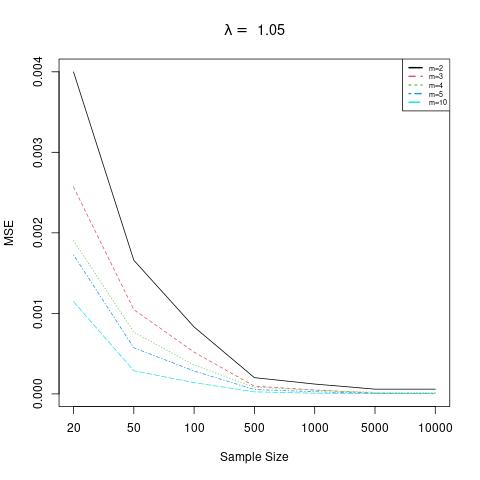}
     \caption{$\lambda=1.05$}
         \label{MSE5Unif}
     \end{subfigure}
     \hfill
     \begin{subfigure}[h]{0.3\textwidth}
         \centering \includegraphics[width=\textwidth]{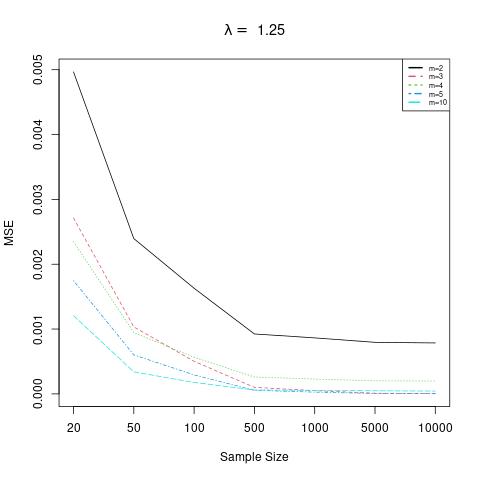} \caption{$\lambda=1.25$}
         \label{MSE6Unif}
     \end{subfigure} \\
          \begin{subfigure}[h]{0.3\textwidth}
         \centering
         \includegraphics[width=\textwidth]{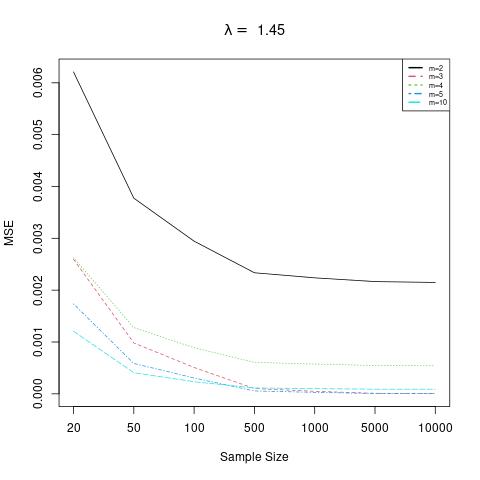}
         \caption{$\lambda=1.45$}
         \label{MSE7unif}
     \end{subfigure}
     \hfill
     \begin{subfigure}[h]{0.3\textwidth}
         \centering
     \includegraphics[width=\textwidth]{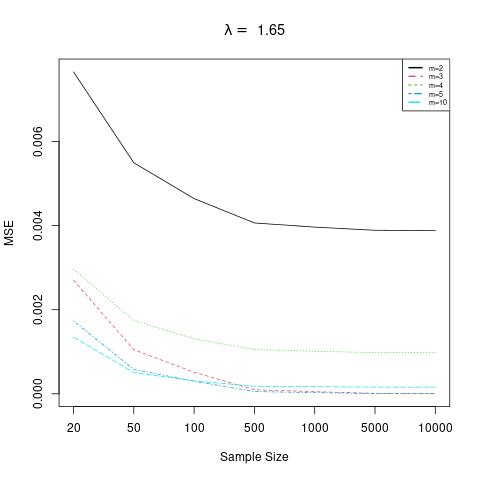}
     \caption{$\lambda=1.65$}
         \label{MSE8unif}
     \end{subfigure}
     \hfill
     \begin{subfigure}[h]{0.3\textwidth}
         \centering \includegraphics[width=\textwidth]{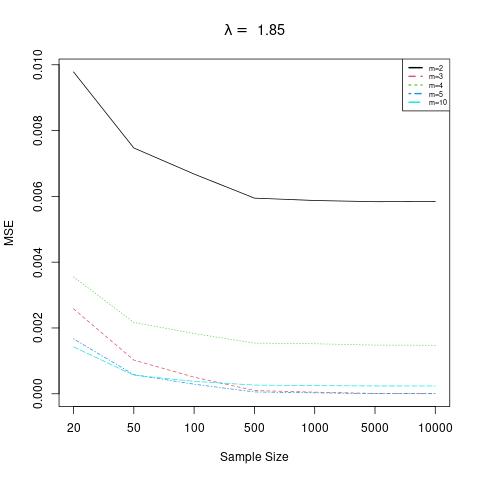}
         \caption{$\lambda=1.85$}
         \label{MSE9unif}
     \end{subfigure}
        \caption{MSE of estimated order quantity for different degrees of severity ($m$) }
    \label{UnifMSEPlot}
\end{figure}

\begin{figure}[h]
     \centering
     \begin{subfigure}[h]{0.3\textwidth}
         \centering
         \includegraphics[width=\textwidth]{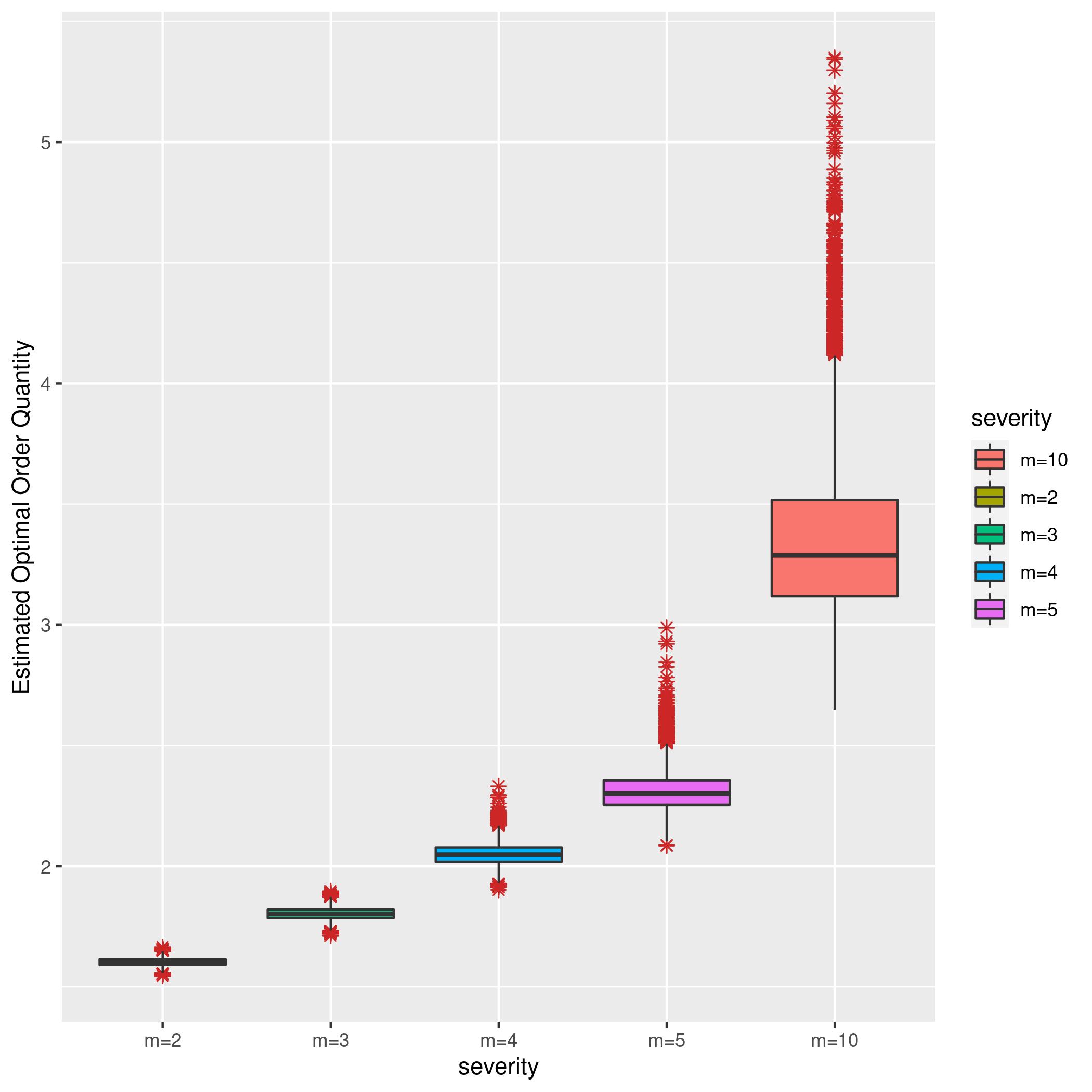}
         \caption{$\lambda=0.25$}
         \label{Qhat1Exp}
     \end{subfigure}
     \hfill
     \begin{subfigure}[h]{0.3\textwidth}
         \centering
     \includegraphics[width=\textwidth]{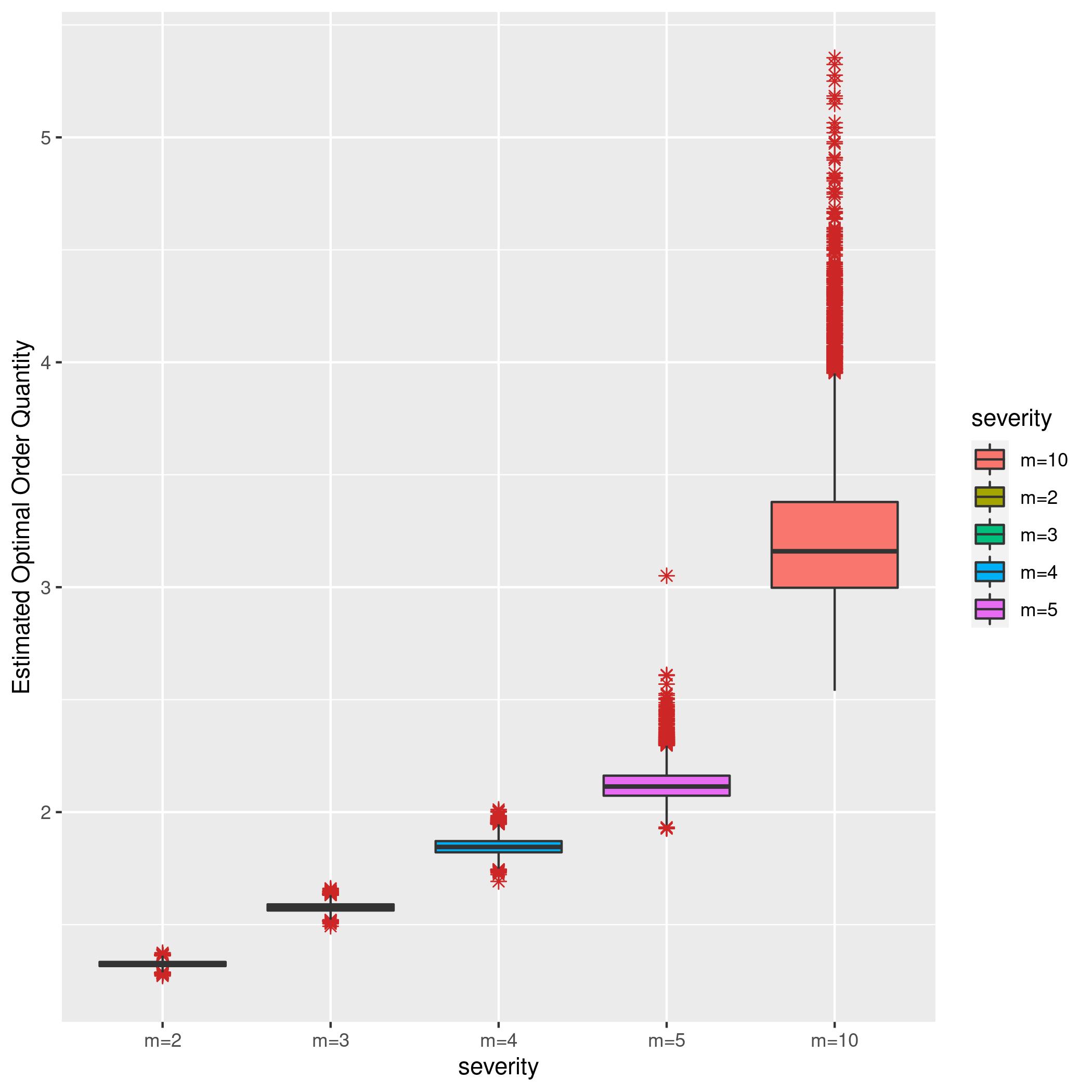}
     \caption{$\lambda=0.45$}
         \label{Qhat2Exp}
     \end{subfigure}
     \hfill
     \begin{subfigure}[h]{0.3\textwidth}
         \centering \includegraphics[width=\textwidth]{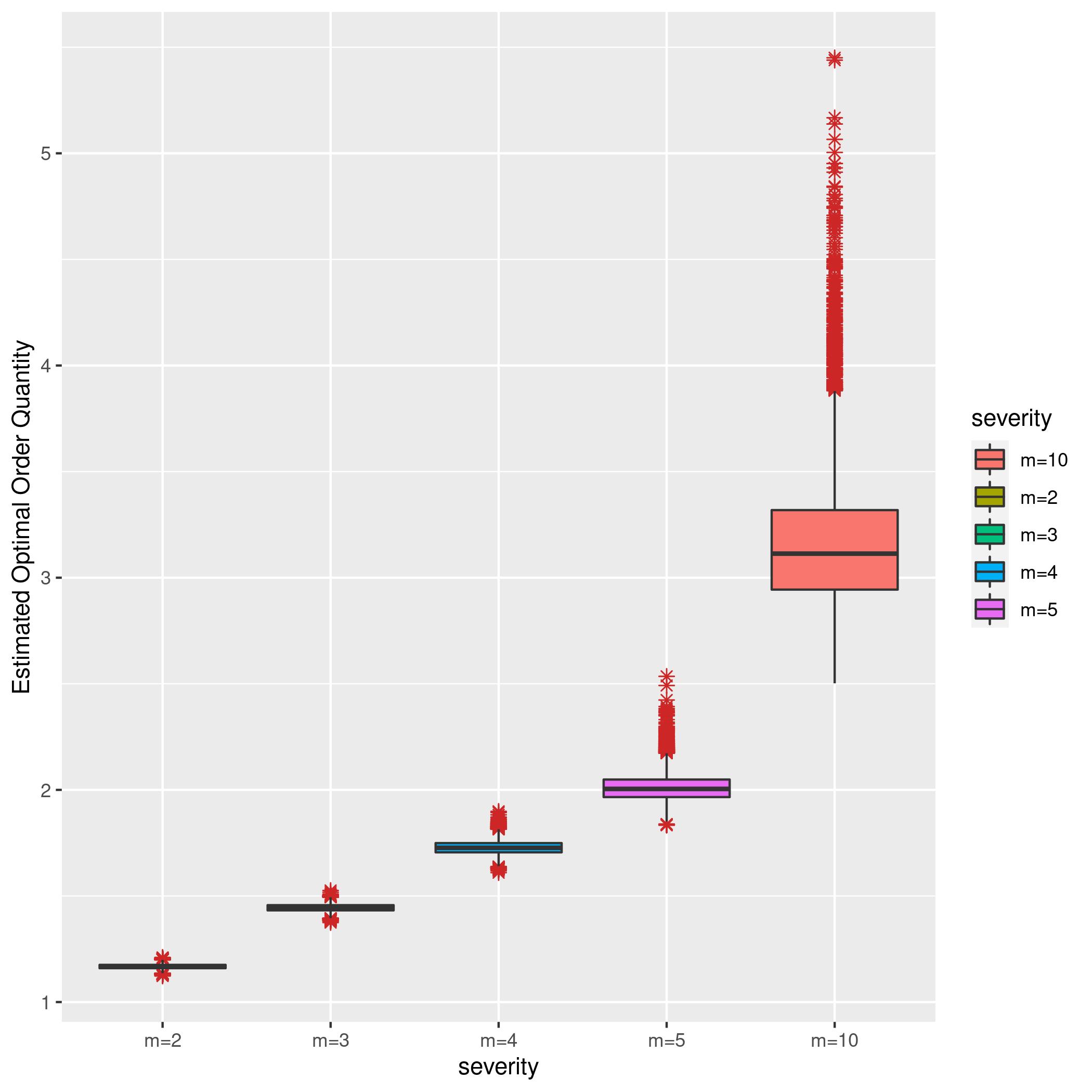}
         \caption{$\lambda=0.65$}
         \label{Qhat3Exp}
     \end{subfigure}
     \\
     \begin{subfigure}[h]{0.3\textwidth}
         \centering
         \includegraphics[width=\textwidth]{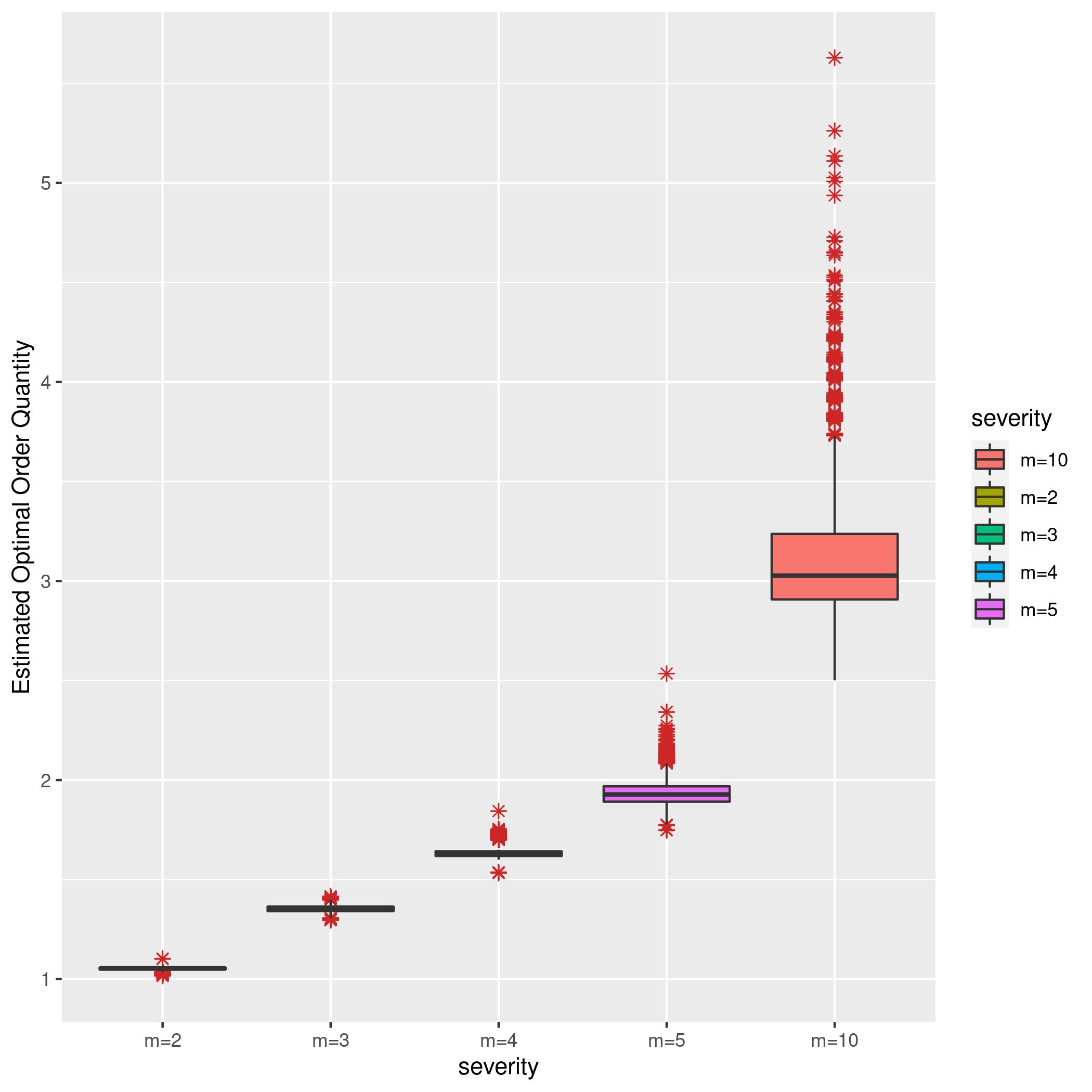}
         \caption{$\lambda=0.85$}
         \label{Qhat4Exp}
     \end{subfigure}
     \hfill
     \begin{subfigure}[h]{0.3\textwidth}
         \centering
     \includegraphics[width=\textwidth]{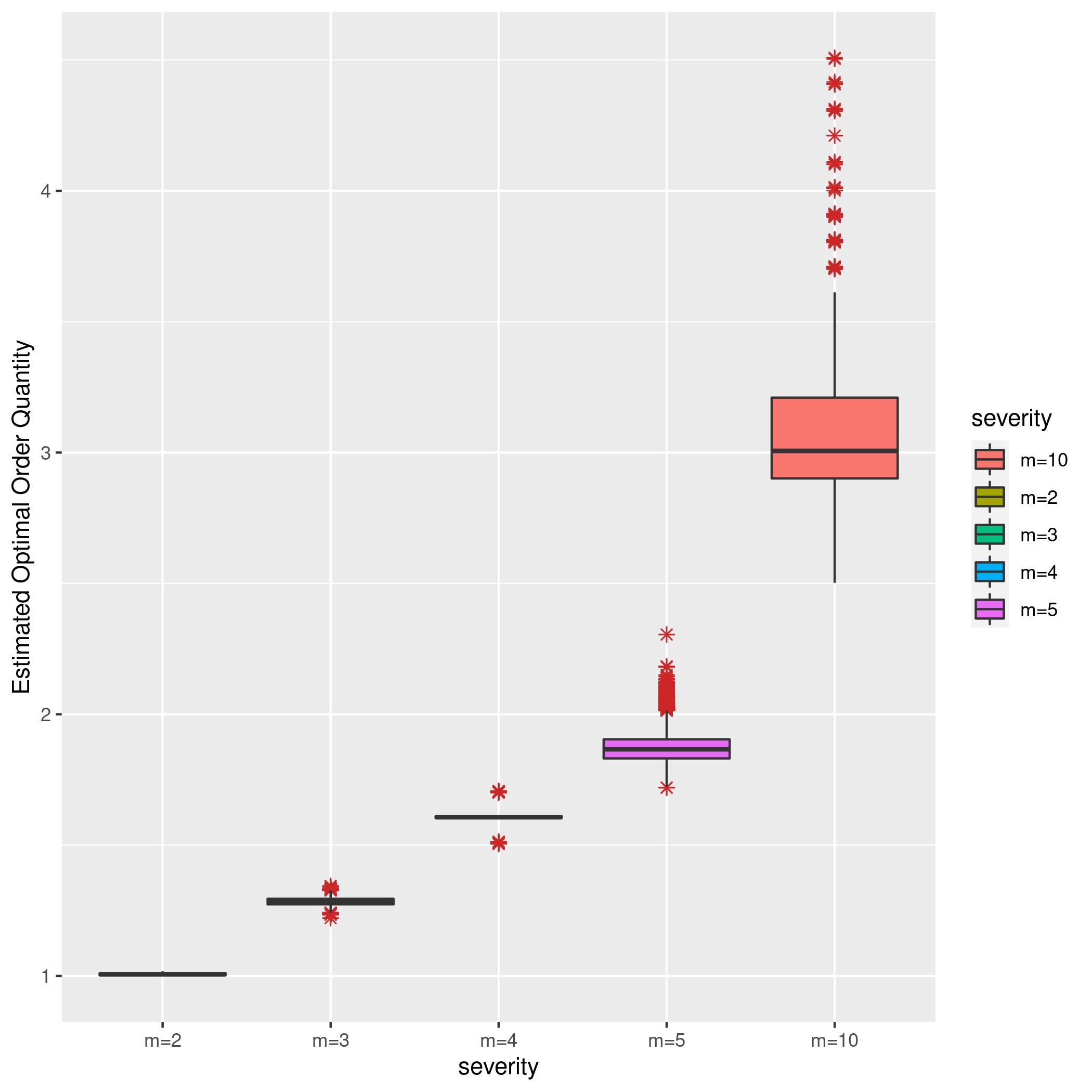}
     \caption{$\lambda=1.05$}
         \label{Qhat5Exp}
     \end{subfigure}
     \hfill
     \begin{subfigure}[h]{0.3\textwidth}
         \centering \includegraphics[width=\textwidth]{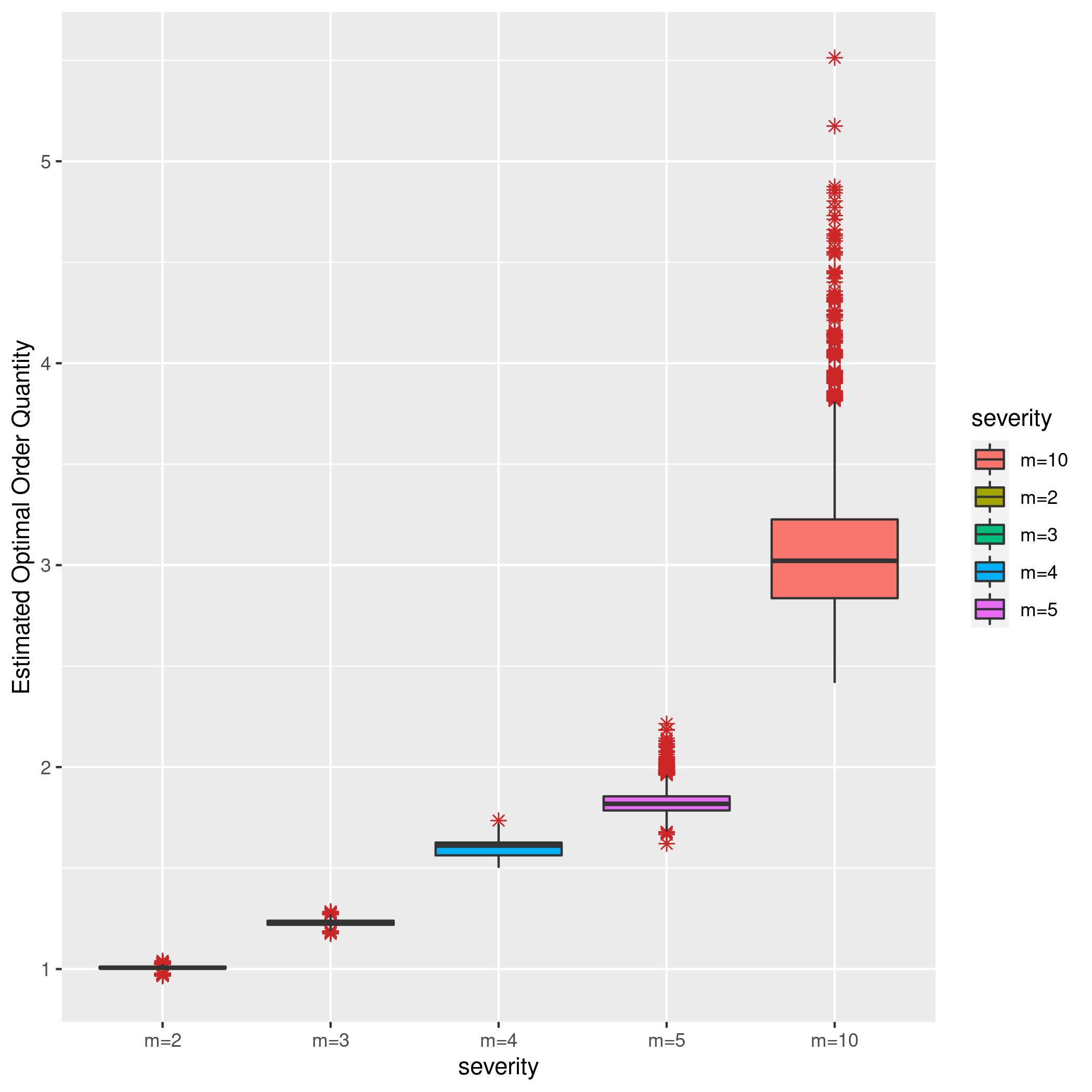}
         \caption{$\lambda=1.25$}
         \label{Qhat6Exp}
     \end{subfigure} \\
          \begin{subfigure}[h]{0.3\textwidth}
         \centering
         \includegraphics[width=\textwidth]{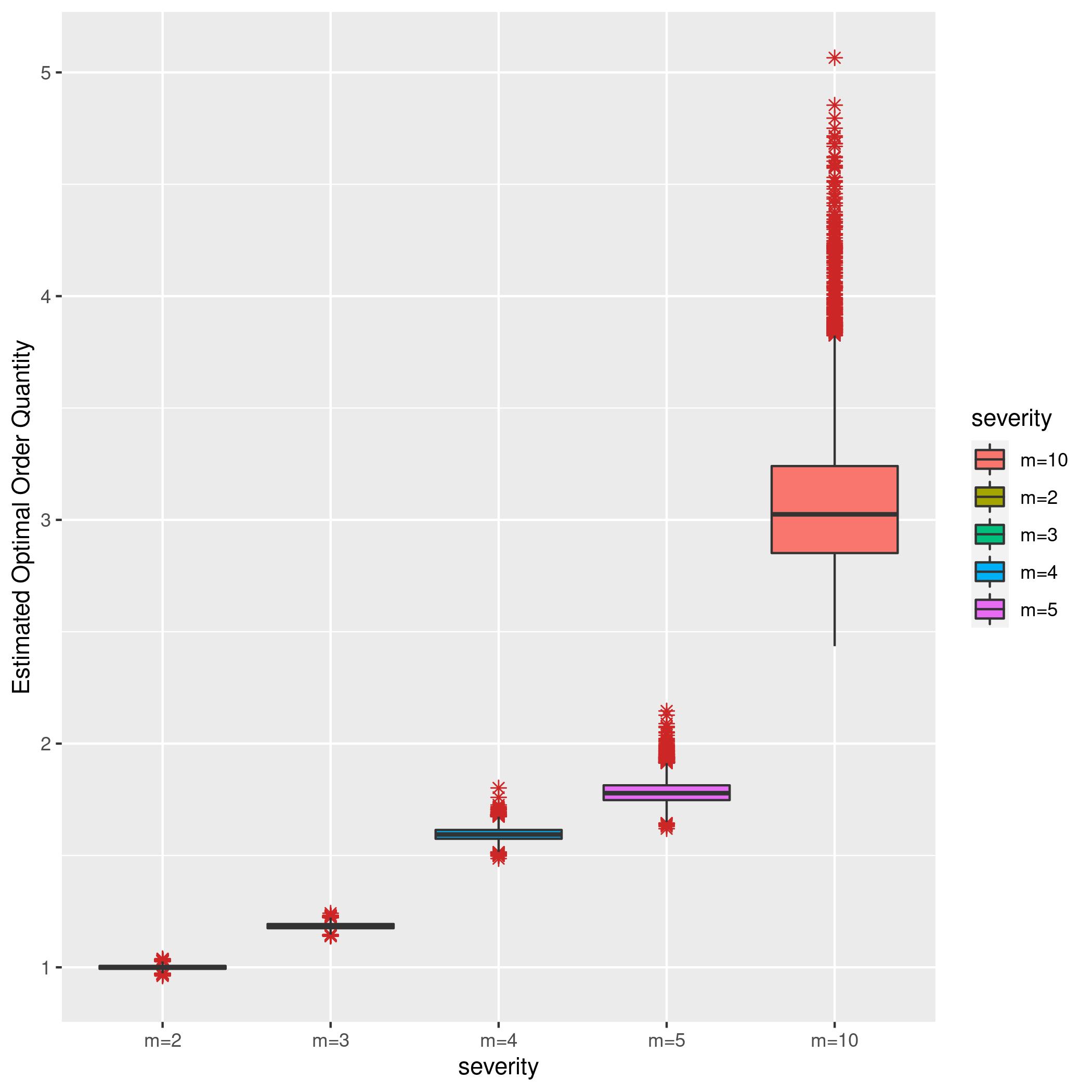}
         \caption{$\lambda=1.45$}
         \label{Qhat7Exp}
     \end{subfigure}
     \hfill
     \begin{subfigure}[h]{0.3\textwidth}
         \centering
     \includegraphics[width=\textwidth]{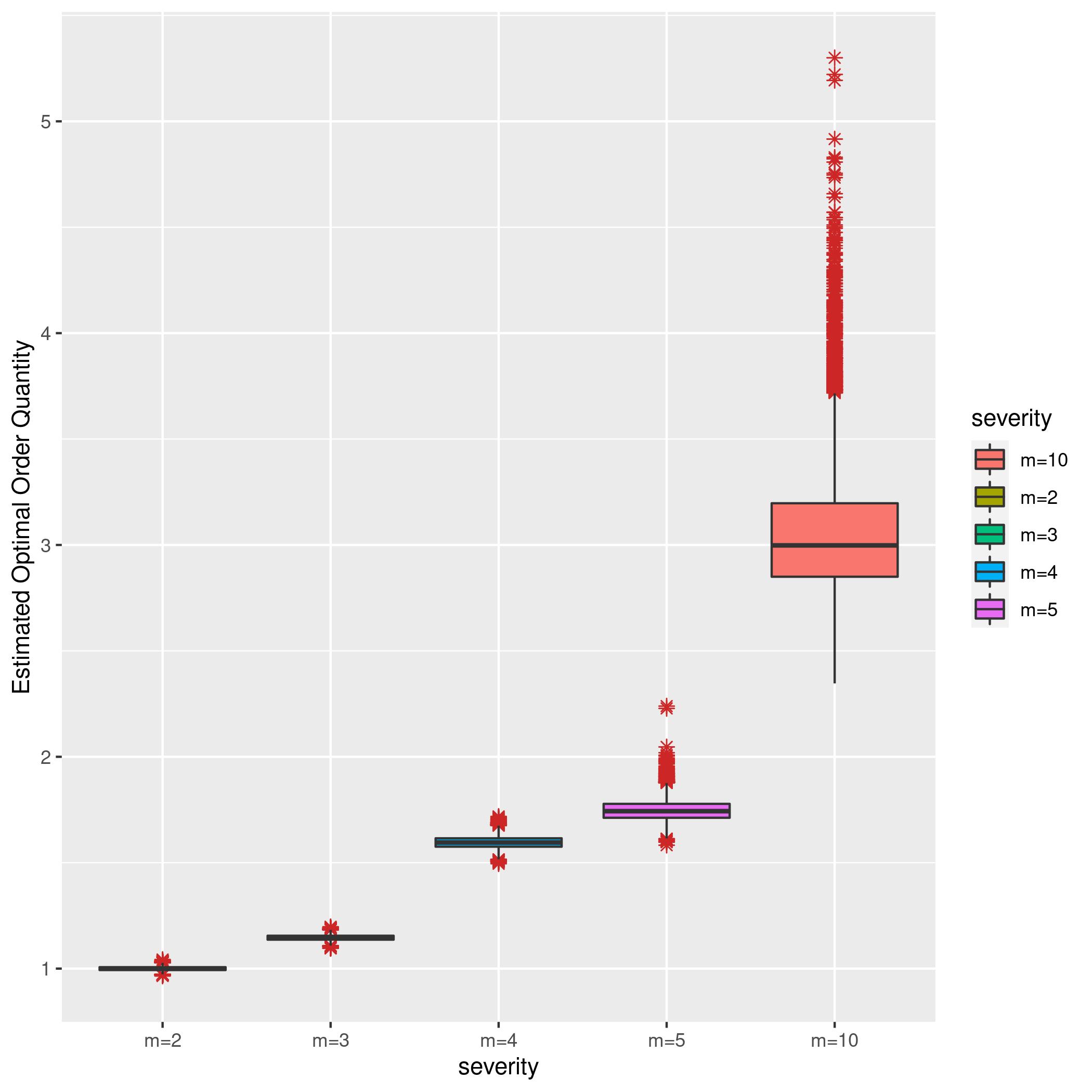}
     \caption{$\lambda=1.65$}
         \label{Qhat8Exp}
     \end{subfigure}
     \hfill
     \begin{subfigure}[h]{0.3\textwidth}
         \centering \includegraphics[width=\textwidth]{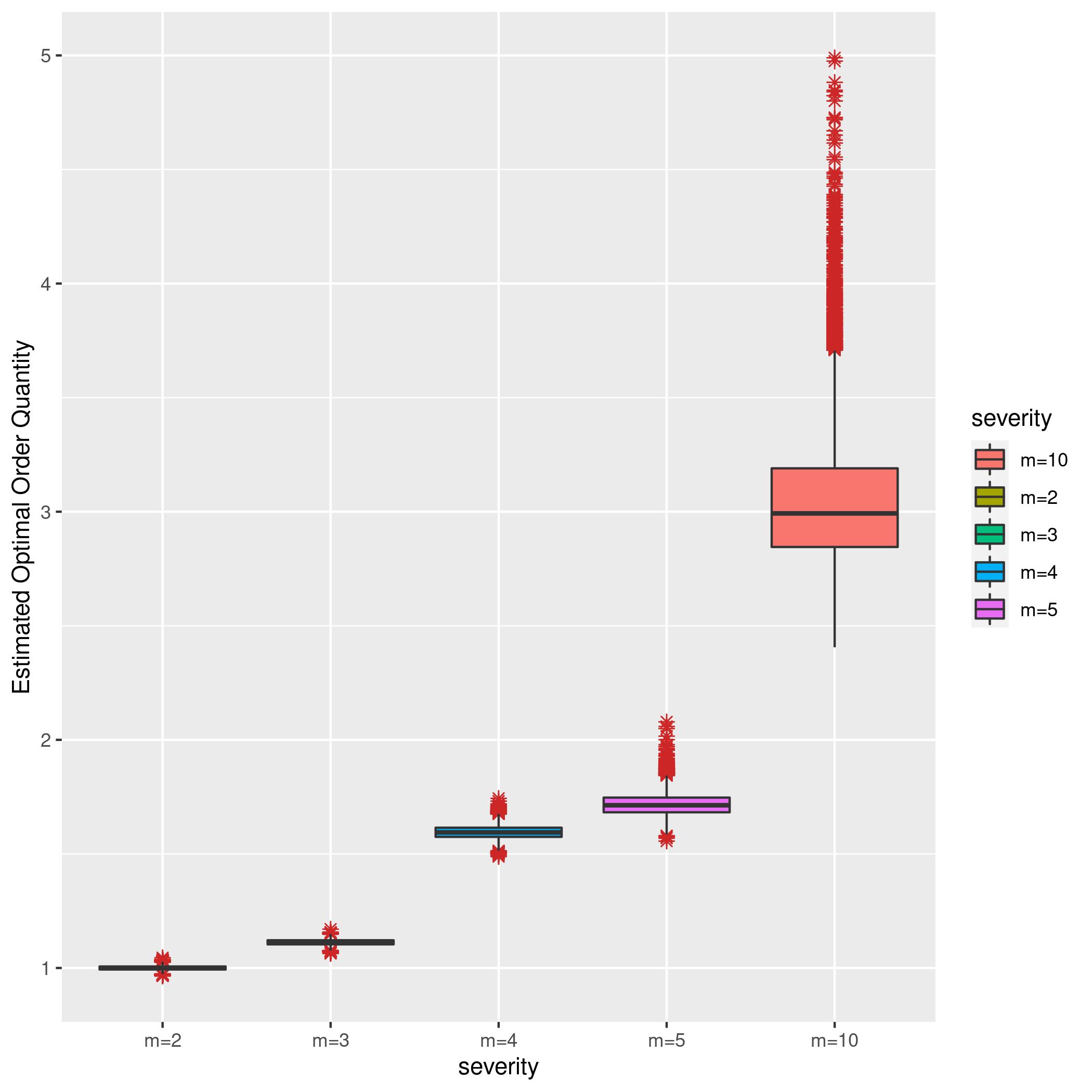}
         \caption{$\lambda=1.85$}
         \label{Qhat9Exp}
     \end{subfigure}
        \caption{Boxplot of estimated order quantity for different degrees of severity ($m$) }
    \label{ExpDenPlot}
\end{figure}
\begin{figure}[h]
     \centering
     \begin{subfigure}[h]{0.3\textwidth}
         \centering
         \includegraphics[width=\textwidth]{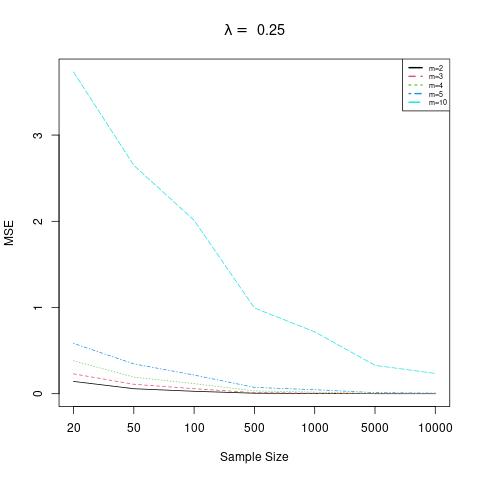}
         \caption{$\lambda=0.25$}
         \label{MSE1Exp}
     \end{subfigure}
     \hfill
     \begin{subfigure}[h]{0.3\textwidth}
         \centering
     \includegraphics[width=\textwidth]{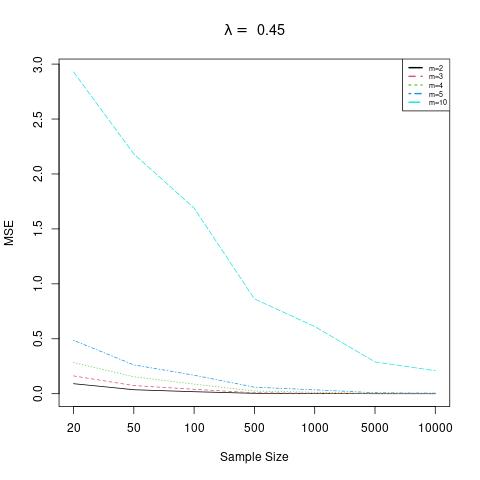}
     \caption{$\lambda=0.45$}
         \label{MSE2Exp}
     \end{subfigure}
     \hfill
     \begin{subfigure}[h]{0.3\textwidth}
         \centering \includegraphics[width=\textwidth]{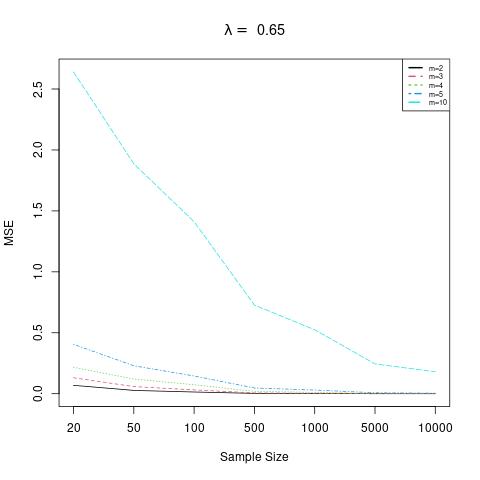}
         \caption{$\lambda=0.65$}
         \label{MSE3Exp}
     \end{subfigure}
     \\
     \begin{subfigure}[h]{0.3\textwidth}
         \centering
         \includegraphics[width=\textwidth]{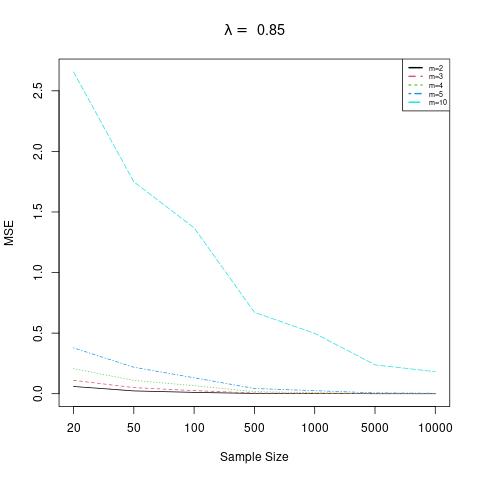}
         \caption{$\lambda=0.85$}
         \label{MSE4Exp}
     \end{subfigure}
     \hfill
     \begin{subfigure}[h]{0.3\textwidth}
         \centering
     \includegraphics[width=\textwidth]{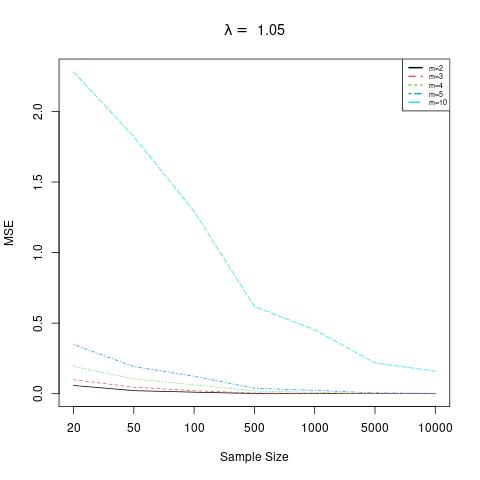}
     \caption{$\lambda=1.05$}
         \label{MSE5Exp}
     \end{subfigure}
     \hfill
     \begin{subfigure}[h]{0.3\textwidth}
         \centering \includegraphics[width=\textwidth]{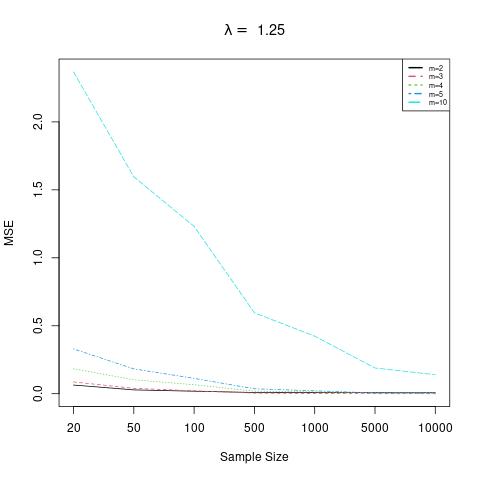} \caption{$\lambda=1.25$}
         \label{MSE6Exp}
     \end{subfigure} \\
          \begin{subfigure}[h]{0.3\textwidth}
         \centering
         \includegraphics[width=\textwidth]{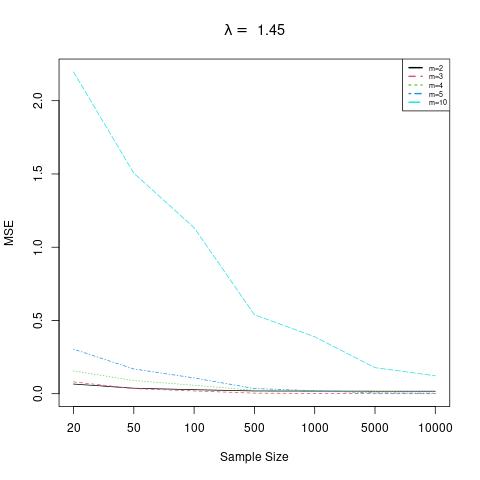}
         \caption{$\lambda=1.45$}
         \label{MSE7Exp}
     \end{subfigure}
     \hfill
     \begin{subfigure}[h]{0.3\textwidth}
         \centering
     \includegraphics[width=\textwidth]{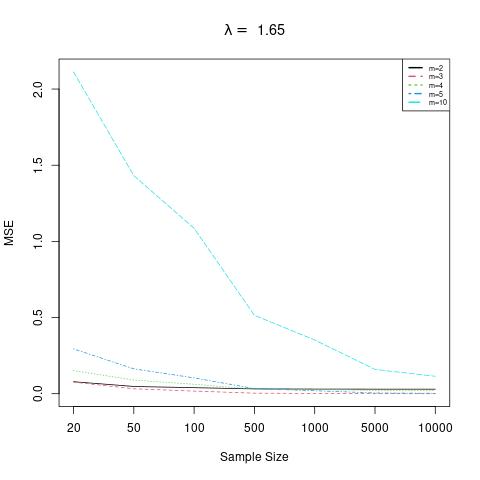}
     \caption{$\lambda=1.65$}
         \label{MSE8Exp}
     \end{subfigure}
     \hfill
     \begin{subfigure}[h]{0.3\textwidth}
         \centering \includegraphics[width=\textwidth]{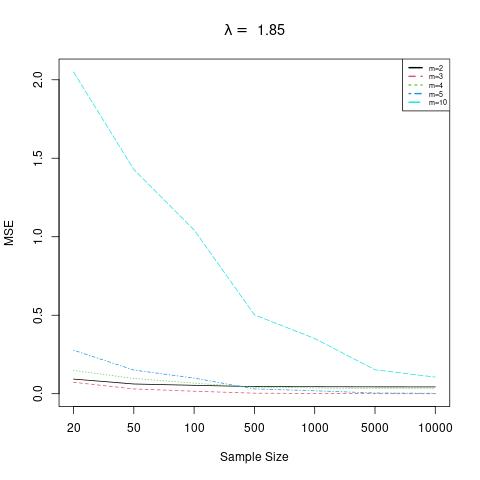}
         \caption{$\lambda=1.85$}
         \label{MSE9Exp}
     \end{subfigure}
        \caption{MSE of estimated order quantity for different degrees of severity ($m$) }
    \label{ExpMSEPlot}
\end{figure}

\end{document}